%% file: main.tex
\documentclass[12pt]{article}
\pdfoutput=1
\usepackage[margin=1.25in]{geometry}
\usepackage[T1]{fontenc}
\usepackage{prelude}

\input{preamble}

\begin{document}



\title[The Gittins Policy's Response Time Tail]{When Does the Gittins Policy \\ Have Asymptotically Optimal \\ Response Time Tail in the M/G/1?}

\author[Ziv Scully and Lucas van Kreveld]{%
  Ziv Scully\\
  Cornell University\\
  \texttt{zivscully@cornell.edu}
  \and
  Lucas van Kreveld\\
  Eindhoven University of Technology\\
  \texttt{l.r.v.kreveld@tue.nl}
}

\date{January 2024}

\maketitle

\bigskip

\begin{abstract}
We consider scheduling in the M/G/1 queue with unknown job sizes. It is known that the Gittins policy minimizes mean response time in this setting. However, the behavior of the tail of response time under Gittins is poorly understood, even in the large-response-time limit. Characterizing Gittins's asymptotic tail behavior is important because if Gittins has optimal tail asymptotics, then it simultaneously provides optimal mean response time and good tail performance.

In this work, we give the first comprehensive account of Gittins's asymptotic tail behavior. For heavy-tailed job sizes, we find that Gittins always has asymptotically optimal tail. The story for light-tailed job sizes is less clear-cut: Gittins's tail can be optimal, pessimal, or in between. To remedy this, we show that a modification of Gittins avoids pessimal tail behavior while achieving near-optimal mean response time.%

\end{abstract}

\section{Introduction}
\label{sec:intro}

Scheduling to minimize response time (a.k.a. sojourn time)
of single-server queueing models is an important problem in queueing theory,
with applications in computer systems, service systems, and beyond.
In general, a queueing system will have a response time \emph{distribution}, denoted~$T$,
and there are a variety of metrics one might hope to minimize.
There is significant work on minimizing \emph{mean response time}~$\E{T}$,
which is the average response time of all jobs in a long arrival sequence
\citep{schrage_proof_1968, gittins_multi-armed_1989, gittins_multi-armed_2011, aalto_gittins_2009}.

Much less is known about minimizing the \emph{tail of response time}~$\P{T > t}$,
which is the probability a job has response time greater than a parameter $t \geq 0$.
In light of the difficulty of studying the tail directly,
theorists have studied the \emph{asymptotic tail of response time},
which is the asymptotic decay of $\P{T > t}$ in the $t \to \infty$ limit
\citep{stolyar_largest_2001, nunez-queija_queues_2002, borst_impact_2003, boxma_tails_2007, scully_characterizing_2020}.
In this work, we consider the preemptive M/G/1 queue and ask the following question.

\begin{restatable:question}
    \label{que:simultaneous}
    Does any scheduling policy \emph{simultaneously} optimize
    the mean and asymptotic tail of response time in the M/G/1?
\end{restatable:question}

Our focus on the M/G/1, a classic single-server queueing model \citep{cox_queues_1961, kendall_stochastic_1953, kleinrock_queueing_1976, harchol-balter_performance_2013}, is motivated by its balance between modeling flexibility and analytical tractability. The fact that the distribution of job sizes (a.k.a. service times) may be general is particularly important for modeling computer systems, where the distribution can be far from exponential \citep{crovella_self-similarity_1997, peterson_data_1996, harchol-balter_exploiting_1997}.

Prior work answers \cref{que:simultaneous}
when job sizes are known to the scheduler.
In this setting, the \emph{Shortest Remaining Processing Time} (SRPT) policy,
which preemptively serves the job of least remaining size,
always minimizes mean response time \citep{schrage_proof_1968}.
However, SRPT's tail performance depends on the job size distribution.\footnote{%
    Although we are referring to the job size distribution, we should clarify that here we are still discussing scheduling with known job sizes, specifically SRPT. But starting shortly, we will shift attention to the case where job sizes are unknown, but we still know the job size distribution from which job sizes are drawn.}
\begin{itemize}
    \item If the job size distribution is \emph{heavy-tailed} (roughly, power-law; see \cref{def:heavy}),
    then SRPT is \emph{tail-optimal},
    meaning it has the best possible asymptotic tail decay (\cref{def:heavy_optimality}).
    \item If the job size distribution is \emph{light-tailed} (roughly, subexponential; see \cref{def:light}),
    then SRPT is \emph{tail-pessimal},
    meaning it has the worst possible asymptotic tail decay (\cref{def:light_optimality}).
\end{itemize}
This answers \cref{que:simultaneous} for known job sizes:
``yes, namely SRPT'' in the heavy-tailed case,
``no'' in the light-tailed case.

Unfortunately, in practice, the scheduler often does not know job sizes,
and thus one cannot implement SRPT.
Instead, the scheduler often only knows the job size \emph{distribution}.
We study \cref{que:simultaneous} in this unknown-size setting.

The question of minimizing mean response time with unknown job sizes
was settled by \citet{gittins_multi-armed_1989}.
He introduced a policy, now known as the \emph{Gittins} policy,
which leverages the job size distribution to minimize mean response time.
Roughly speaking, Gittins uses each job's \emph{age},
namely the amount of time each job has been served so far,
to figure out which job is most likely to complete after a small amount of service,
then serves that job.
For some job size distributions, Gittins reduces to a simpler policy,
such as \emph{First-Come, First-Served} (FCFS) or \emph{Foreground-Background} (FB)
\citep{aalto_gittins_2009, aalto_properties_2011}.

In the unknown-size setting,
given that Gittins minimizes mean response time,
\cref{que:simultaneous} reduces to the following.

\begin{restatable:question}
    \label{que:gittins_tail}
    For which job size distributions is Gittins tail-optimal for response time?
\end{restatable:question}

Unfortunately, the asymptotic tail behavior of Gittins is understood in only a few special cases.
\begin{itemize}
    \item In the heavy-tailed case,
    Gittins has been shown to be tail-optimal,
    but only under an assumption on the job size distribution's hazard rate
    \citep[Corollary~3.5]{scully_characterizing_2020}.
    \item In the light-tailed case,
    Gittins sometimes reduces to FCFS or FB
    \citep{aalto_gittins_2009, aalto_properties_2011}.
    For light-tailed job sizes,
    FCFS is tail-optimal \citep{stolyar_largest_2001, boxma_tails_2007},
    but FB is tail-pessimal \citep{mandjes_sojourn_2005}.
\end{itemize}
This prior work leaves \cref{que:gittins_tail} largely open.
We do not know whether Gittins is always tail-optimal in the heavy-tailed case,
or whether it is sometimes suboptimal, or even tail-pessimal.
Moreover, we do not understand Gittins's asymptotic tail at all in the light-tailed case, aside from when Gittins happens to reduce to a simpler policy.

The prior work above does tell us an important fact:
Gittins \emph{can} be tail-pessimal.
This prompts another question.

\begin{restatable:question}
    \label{que:improve_pessimal}
    For job size distributions for which Gittins is tail-pessimal,
    is there another policy that has \emph{near-optimal} mean response time
    while not being tail-pessimal?
\end{restatable:question}

In this work, we answer \cref{que:simultaneous, que:gittins_tail, que:improve_pessimal}
for the M/G/1 with unknown job sizes,
covering wide classes of heavy- and light-tailed job size distributions.
The key tool we use to analyze Gittins's asymptotic response time tail
is the \emph{SOAP} framework \citep{scully_soap_2018, scully_characterizing_2020}.
SOAP gives a universal M/G/1 response time analysis
of all \emph{SOAP policies},
which are scheduling policies where a job's priority level is a function of its age
(\cref{def:soap}).
Underlying our Gittins results is a general tail analysis of SOAP policies.

Our main contributions,
which we describe in more detail later (\cref{sec:heavy,sec:light}), are as follows:
\begin{itemize}
    \item \emph{Heavy-tailed case:} We give a sufficient condition under which an arbitrary SOAP policy
    is tail-optimal (\cref{sec:heavy_soap}).
    \item \emph{Heavy-tailed case:} We show that the above condition always applies to Gittins,
    implying it is always tail-optimal (\cref{sec:heavy_gittins}).
    \item \emph{Light-tailed case:} We characterize when an arbitrary SOAP policy
    is tail-optimal, tail-pessimal, or in between (\cref{sec:light_soap}).
    \item \emph{Light-tailed case:} We spell out how the above characterization applies to Gittins and show how to modify Gittins to avoid tail pessimality
    (\cref{sec:light_gittins}).
    \item \emph{General case:} At the core of our modification of Gittins
    which avoids tail pessimality
    is a general result which states that
    slightly perturbing the Gittins rank function
    only slightly affects its mean response time
    (\cref{thm:approximate_gittins, sec:gittins}).\footnote{%
        During the review process, a generalization of this result was shown by \citet[Chapter~16]{scully_new_2022}.}
\end{itemize}
The rest of the paper
introduces definitions and notation (\cref{sec:preliminaries, sec:soap}),
and concludes with some remarks about our motivating questions (\cref{sec:conclusion}).

\section{Prior Work}
\label{sec:prior_work}

\subsection{Asymptotic Tail Analysis of Classic Scheduling Policies}
\label{sec:prior_work:tail_asymptotics}

Due to their frequent occurrence, light-tailed job size distributions have received a great amount of attention by queueing theorists. The performance of policies under light-tailed job sizes is generally measured in terms of the decay rate of the response time tail. In this sense FCFS has proven to be optimal among all service policies \cite{stolyar_largest_2001}. Conversely, \emph{Foreground-Background Processor Sharing} (FB) has the worst possible decay rate of the response time tail \cite{mandjes_sojourn_2005}.

On the other hand, it is shown that heavy-tailed job sizes can have a large impact on the performance characteristics of the queue. For this reason also heavy-tailed job sizes have been thoroughly investigated in the literature. For example, researches have made the striking observation that, contrary to the light-tailed case, FB is optimal where FCFS has the worst possible response time tail \cite{borst_impact_2003}. This dichotomy between light and heavy tails is not limited to FCFS and FB \cite{boxma_tails_2007}.

Other noteworthy literature highlighting both light and heavy tails includes delicate asymptotic results for a two-class priority policy \cite{abate_asymptotics_1997} and robust optimization using a limited PS policy \cite{nair_tail-robust_2010}.

With one exception, discussed in \cref{sec:sec:prior_work:heavy_soap}, the literature on this subject concerns only a few relatively simple policies. 
This paper considers policies in which the priority of a job can vary essentially arbitrarily with its age. This generality is needed to analyze the Gittins policy, where a job's priority can be non-monotonic \citep{aalto_properties_2011}.

\subsection{Impossibility of Universal Tail-Optimal Scheduling}
\label{sec:prior_work:impossibility}

In light of the fact that both FCFS and FB can vary between tail-optimal and tail-pessimal for different job sizes, it is natural to ask whether there is a single policy that is always tail-optimal. \Citet{wierman_is_2012} answer this question with an impossibility result, showing that no policy is tail-optimal for both heavy- and light-tailed job sizes, unless the policy has knowledge of the size distribution~$X$ or learns~$X$ over time.

One might worry that this impossibility result contradicts our results for Gittins, given that we show Gittins is always tail-optimal in the heavy-tailed case and sometimes tail-optimal in the light-tailed case. The reason there is no contradiction is that the Gittins policy changes based on the size distribution. For instance, there are some light-tailed distributions where Gittins reduces to FCFS, and some heavy-tailed distributions where Gittins reduces to FB \citep{aalto_gittins_2009}.

\subsection{Tail Optimality of Certain SOAP Policies in the Heavy-Tailed Case}
\label{sec:sec:prior_work:heavy_soap}

We mention particularly the relation between this paper and the work of \citet{scully_characterizing_2020}.
Both this paper and the prior work study the response time tail behavior of arbitrary SOAP policies, including the Gittins policy.
There are two main factors that distinguish this paper from the prior work.
\begin{itemize}
\item
    \Citet{scully_characterizing_2020} only study heavy-tailed job size distributions.
    In contrast, we study both the heavy- and light-tailed cases.
\item
     \Citet{scully_characterizing_2020} show that Gittins is tail-optimal subject to a condition on the job size distribution's hazard rate \citep[Corollary~3.5]{scully_characterizing_2020}. However, their analysis is not sharp enough to completely characterize under which (heavy-tailed) job size distributions Gittins is tail-optimal. In contrast, our analysis is sharper, allowing us to identify Gittins's tail performance under any job size distribution.
\end{itemize}
With this said,
\citet{scully_characterizing_2020} lay an important technical foundation
that we build upon to derive our heavy-tailed results.
See \cref{rmk:prior_work_comparison} for a more technical discussion
of what aspects of their work we use and what aspects we improve upon.


\subsection{Beyond Asymptotic Tail Optimality}

It is well known that FCFS has optimal tail decay rate under light-tailed job sizes. However, decay rate is a relatively crude tail performance measure, as it does not take into account the constant (or non-exponential term) in front of the the exponent. Although this paper focuses just on decay rates, we mention that very recently a policy was introduced that has a better leading constant than FCFS \cite{grosof_nudge_2021}. An open question remains what is the best possible leading constant in the response time tail. A by-product of our results, namely that FCFS is the only SOAP policy with optimal decay rate, partially answers this question. Specifically, it follows that no SOAP policy is tail-optimal up to the leading constant.


\subsection{Mean Response Time of Modified Gittins Policies}
\label{sec:prior_work:approximate_gittins}


A recent study \citep[Theorem~7.2]{scully_uniform_2022} shows that
if one slightly modifies the prioritization rules of SRPT,
then the mean response time of the resulting policy
is only slightly worse than that of unmodified SRPT (which is optimal in case job sizes are known). It turns out, as shown in this paper, that essentially the same result holds for an approximate version of the Gittins policy, which can thus be seen as the unknown-job-sizes counterpart of \citep[Theorem~7.2]{scully_uniform_2022}.\footnote{%
    In fact, both results are special cases of a more general result. See \citet[Chapter~16]{scully_new_2022} for details.}

\section{Model, SOAP Policies, and the Gittins Policy}
\label{sec:preliminaries}

We consider an M/G/1 queue with arrival rate~$\lambda$, job size distribution~$X$, and load $\rho = \lambda \E{X}$.
For the tail of the job size distribution, we write $\tail{t} = \P{X > t}$.
We denote the maximum job size by $\xmax = \inf\curlgp{t \geq 0 \given \tail{t} = 0}$,
allowing $\xmax = \infty$.
We write $T_\pi$ for the M/G/1's response time distribution under policy~$\pi$.
We allow policies to preempt jobs or share the processor without any overhead or loss of work
(i.e. the model is preempt-resume).

Special attention is given in this paper to the Gittins policy. It assigns each job a \emph{rank}, namely a priority,
based on the job's \emph{age},
namely the amount of time the job has been served so far.
To analyze the Gittins policy, we make use of the \emph{SOAP framework}
\citep{scully_soap_2018, scully_characterizing_2020},
which gives a response time analysis of the following broad class of policies.

\begin{definition}
    \label{def:soap}
    A \emph{SOAP policy} is a policy~$\pi$ specified by a \emph{rank function}
    $r_\pi : [0, \xmax) \to \R$.
    Policy~$\pi$ assigns rank $r_\pi(a)$ to a job at age~$a$.\footnote{%
        The full SOAP definition is more general \citep{scully_soap_2018},
        but the given definition suffices for our unknown-size setting.}
    When the policy being discussed is clear from context,
    we often omit the subscript and simply write~$r(a)$.
    At every moment in time, a SOAP policy \emph{serves the job of minimum rank},
    breaking ties in FCFS order.\footnote{%
        We give some brief remarks on other tiebreaking rules in \cref{sec:conclusion:questions}.}
\end{definition}

\begin{definition}
    \label{def:gittins}
    The \emph{Gittins} policy, denoted ``$\gittins$'' in subscripts for brevity,
    is the SOAP policy with rank function
    \begin{align*}
        r_\gittins(a)
        &= \inf_{b > a} \frac{\E{\min\{S, b\} - a \given S > a}}{\P{S \leq b \given S > a}}
        \iftwocoleq{\\ &}{}
        = \inf_{b > a} \frac{\int_a^b \tail{t} \d{t}}{\tail{a} - \tail{b}}.
    \end{align*}
    Note that the Gittins rank function depends on the job size distribution~$X$ by way of~$\tail{}$.
\end{definition}

\label{asm:nice_dist}
As is standard \citep[Appendix~B]{scully_soap_2018}, we assume rank functions are \emph{piecewise-continuous} and \emph{piecewise-monotonic},
with finitely many pieces in any compact interval for both properties.
This holds for Gittins under very mild conditions on the job size distribution~$X$.
For example, \citet{aalto_properties_2011} show the Gittins rank function
is continuous and piecewise-monotonic provided that $X$ is a continuous distribution
with continuous and piecewise-monotonic hazard rate.
However, our results are not restricted to continuous job size distributions.
Our generic SOAP results require no additional assumptions on~$X$,
and our Gittins results require only that $X$ induces
a piecewise-continuous and piecewise-monotonic Gittins rank function,
which can occur even if $X$ is not continuous.

\section{Heavy-Tailed Job Sizes}\label{sec:heavy}


In \cref{sec:preliminaries:heavy}, we define which job size distributions are heavy-tailed and we give our criterion for tail optimality in this scenario. The two main results in the heavy-tailed case are presented in \cref{sec:results:heavy}:
\begin{itemize}
    \item \Cref{thm:heavy_soap} gives a sufficient condition for
    a SOAP policy to be tail-optimal for heavy-tailed job sizes.
    \item \Cref{thm:heavy_gittins} shows that for heavy-tailed job sizes,
    Gittins always satisfies this sufficient condition,
    and is thus always tail-optimal.
\end{itemize}

\subsection{Background on Heavy-Tailed Job Sizes}
\label{sec:preliminaries:heavy}

Roughly speaking, the heavy-tailed job size distributions we study
are those which are asymptotically Pareto.
The specific class we study, described below,
is slightly more general in that it also includes distributions
whose tails oscillate between Pareto tails of different shape parameters.

\begin{definition}[Heavy-Tailed Job Size Distribution]
    \label{def:heavy}
    We say a job size distribution $X$ is \emph{nicely heavy-tailed} if $\xmax = \infty$ and both of the following hold:
    \begin{enumerate:definition}
    \item
        The tail~$\tail{\cdot}$ is of intermediate regular variation \citep{cline_intermediate_1994},
        meaning
        \begin{equation*}
            \liminf_{\epsilon \downarrow 0} \liminf_{x \to \infty} \frac{\tail{(1 + \epsilon)x}}{\tail{x}} = 1.
        \end{equation*}
    \item
        There exist $\beta \geq \alpha > 1$ such that the upper and lower Matuszewska indices of $\tail{\cdot}$ are in $(-\beta, -\alpha)$ \citep[Section~2.1]{bingham_regular_1987}.
        This implies that for all sufficiently large $x_2 \geq x_1$,\footnote{%
            \label{fn:asymptotic_notation}%
            We formally define the $O(\cdot)$, $\Omega(\cdot)$, and $\Theta(\cdot)$ notations as follows.
            Suppose $x_1, \dots, x_n$ are non-negative variables.
            The notation $O(f(x_1, \dots, x_n))$ stands for an unspecified expression $g(x_1, \dots, x_n) \geq 0$
            for which there exist constants $C, y_0, \dots, y_n > 0$ such that
            for all $x_1 \geq y_1, \dots, x_n \geq y_n$,
            we have $g(x_1, \dots, x_n) \leq C f(x_1, \dots, x_n)$.
            The $\Omega(\cdot)$ notation is the same but with the inequality reversed,
            and the $\Theta(\cdot)$ notation indicates that both inequalities hold, likely with different values for the constants.
            For all of these notations, the constants may depend on the job size distribution~$X$.}
        \begin{equation*}
            \Omega\gp[\bigg]{\gp[\bigg]{\frac{x_2}{x_1}}^{-\beta}}
            \leq \frac{\tail{x_2}}{\tail{x_1}}
            \leq O\gp[\bigg]{\gp[\bigg]{\frac{x_2}{x_1}}^{-\alpha}}.
        \end{equation*}
    \end{enumerate:definition}
    In informal discussion, we omit ``nicely''. Although the above definition includes all distributions with power-law-like tail decay, it is noted that we do not consider tails of a less heavy order, such as those of the lognormal and Weibull distributions.
\end{definition}

\begin{definition}[Tail Optimality in Heavy-Tailed Case]
    \label{def:heavy_optimality}
    Consider an M/G/1 with nicely heavy-tailed job size distribution~$X$.
    We call a scheduling policy~$\pi$ \emph{tail-optimal}
    among preemptive work-conserving policies if
    \begin{equation*}
        \lim_{t \to \infty} \frac{\P{T_\pi > t}}{\tail{(1 - \rho)t}} = 1.
    \end{equation*}
\end{definition}

That is, tail optimality holds
if large jobs have a response time of approximately $1/(1-\rho)$ times their size,
which is the best possible asymptotic tail decay in the heavy-tailed case
\citep{wierman_is_2012, boxma_tails_2007}.

\subsection{Results for Heavy-Tailed Case}
\label{sec:results:heavy}

Let us focus on a tagged job of size~$x$. For determining whether or not it will be delayed by other jobs, it is important to know the worst (highest) ever rank that it will ever have, as well as the ages at which other jobs will have rank lower than that worst ever rank.

\begin{definition}
    \label{def:worst_ever_rank}
    The \emph{worst ever rank} of a job of size~$x$ is defined by $w_x = \sup_{0 \leq a < x} r(a)$.
\end{definition}

\begin{definition}
    \label{def:relevant_age_interval}
    \leavevmode
    \begin{enumerate:definition}[beginpenalty=10000]
    \item
        A \emph{$w$-interval} is an interval $(b, c)$ with $0 \leq b < c \leq \xmax$ such that $r(a) \leq w$ for all $a \in (b, c)$.
    \item
        A $w$-interval $(b, c)$ is \emph{right-maximal} if for all $c' \geq c$, the interval $(b, c')$ is not a $w$-interval. This is equivalent to $c$ satisfying either $r(c) \geq w$ or $c = \xmax$. We define \emph{left-maximal} similarly, and we call a $w$-interval \emph{maximal} if it is both left- and right-maximal.
    \end{enumerate:definition}
\end{definition}

Note that the tagged job of size~$x$, no matter its age, always has priority over jobs that have rank higher than~$w_x$. Therefore, it can only be delayed by another job if that other job's age is in a $w_x$-interval.
See \cref{fig:w_x-intervals} for an illustration.
To ensure that the tagged job does not wait too long behind other jobs,
the $w_x$-intervals must be relatively short.
We use the following condition to characterize the length of $w_x$-intervals.

\begin{figure}
    \centering
    \input{figures/w_x-intervals}
    \caption{%
        Illustration of worst ever rank~$w_x$ (purple dotted line) and maximal $w_x$-intervals (orange regions) for a SOAP policy given by rank function~$r$ (cyan curve). A tagged job of size~$x$ always has rank $w_x$ or better, so if another job has priority over the tagged job, that other job's age must be in a $w_x$-interval.}
    \label{fig:w_x-intervals}
\end{figure}
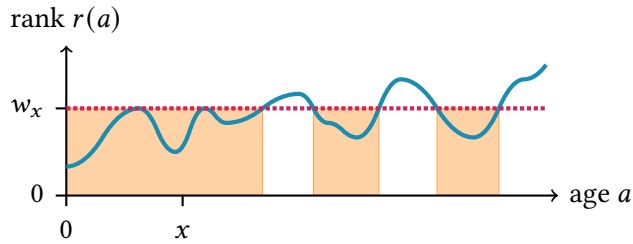

\begin{condition}
    \label{cond:exponents}
    There exist $\zeta, \theta \in [0, \infty)$ and $\eta \in [\max\{1, \zeta + \theta\}, \infty]$ such that the following hold for any $w_x$-interval $(b, c)$ with $b \geq x$:
    \begin{enumerate:condition}
    \item
        \label{cond:exponents_length}
        The $w_x$-interval's length is bounded by $c - b \leq O(b^\zeta x^\theta)$.
    \item
        \label{cond:exponents_max}
        The $w_x$-interval's endpoint is bounded by $c \leq O(x^\eta)$.\footnote{%
            This is trivially satisfied when $\eta = \infty$.}
    \end{enumerate:condition}
    Note that to check that (i) and~(ii) hold, it suffices to consider only right-maximal $w_x$-intervals.
\end{condition}

At an intuitive level, we can think of \cref{cond:exponents} as saying the following about how much other jobs delay the tagged job of size~$x$:
\begin{itemize}
\item
    If $\zeta$ and $\theta$ are small enough, then $w_x$-intervals are relatively short, so jobs of age greater than~$x$ cannot delay the tagged job for too long. \Cref{fig:zeta_vs_theta} illustrates the difference between the roles of $\zeta$ and~$\theta$ (see also \cref{rmk:prior_work_comparison}).
\item
    If $\eta$ is small enough, then there are no $w_x$-intervals at sufficiently large ages, so jobs of sufficiently large age cannot delay the tagged job at all.
\end{itemize}

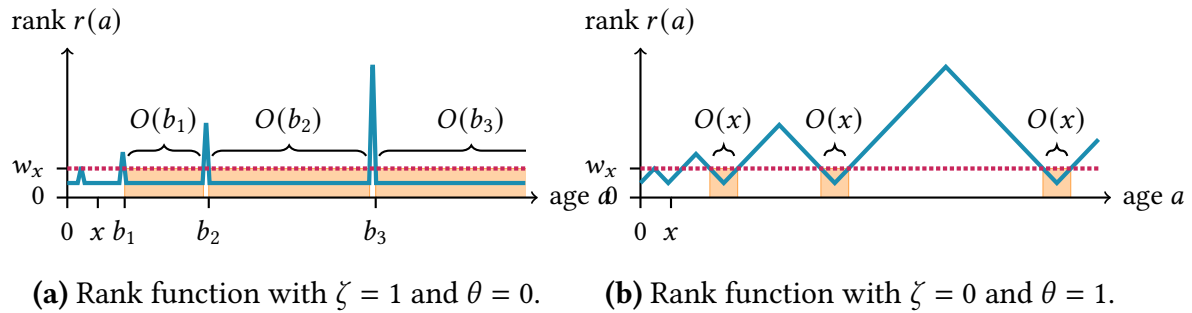
\begin{figure*}
    \renewcommand{\xscale}{9}
    \centering
    \begin{subfigure}[t]{0.5\linewidth}
        \centering
        \input{figures/w_x-intervals_zeta}
        \caption{Rank function with $\zeta = 1$ and $\theta = 0$.}
    \end{subfigure}\hfill
    \begin{subfigure}[t]{0.5\linewidth}
        \centering
        \input{figures/w_x-intervals_theta}
        \caption{Rank function with $\zeta = 0$ and $\theta = 1$.}
    \end{subfigure}
    \caption{%
        Illustration of \cref{cond:exponents}(i), with examples showing the roles $\zeta$ and $\theta$ play. Roughly speaking, one should think of $\zeta + \theta$ as characterizing how far apart different ``peaks'' of the rank function are, and one should think of $\zeta/(\zeta + \theta)$ as characterizing how ``steep'' the rank function is between peaks. Both (a) and (b) have the same peaks, as reflected by (a) and (b) having the same value of $\zeta + \theta$. But the slopes are much steeper in (a) than in~(b), as reflected by the larger value of $\zeta/(\zeta + \theta)$ in (a) than in~(b).}
    \label{fig:zeta_vs_theta}
\end{figure*}

Rank functions that satisfy \cref{cond:exponents} do so for many possible parameter values. For instance, if \cref{cond:exponents} is satisfied with parameters $(\zeta, \theta, \eta)$, then it is also satisfied for $(\zeta + \delta, \theta - \delta + \epsilon, \eta + \epsilon)$ for all $\delta, \epsilon > 0$. But the idea is to find parameters that characterize a SOAP policy's rank function as tightly as possible, because, as our main result below shows, if the parameters are small enough, then the policy is tail-optimal.

\begin{theorem}
    \label{thm:heavy_soap}
    Consider an M/G/1 with any nicely heavy-tailed job size distribution under a SOAP policy.
    \Cref{cond:exponents} implies the policy is tail-optimal if its parameters satisfy
    \begin{equation}
        \label{eq:sufficient}
        \zeta + (\theta - 1)^+ - \frac{(1 - \theta)^+}{\eta} < \frac{\alpha - 1}{\beta}.
    \end{equation}
\end{theorem}

We can apply \cref{thm:heavy_soap} to show that Gittins is tail-optimal. Specifically, we will show that Gittins satisfies \cref{cond:exponents} with $\zeta = 0$, $\theta = 1$, and $\eta = \infty$. As such, it achieves a value of~$0$ on the left-hand side of \cref{eq:sufficient}, so Gittins is tail-optimal regardless of the values of $\beta \geq \alpha > 1$.

\begin{theorem}
    \label{thm:heavy_gittins}
    The Gittins policy is tail-optimal for any nicely heavy-tailed job size distribution.
\end{theorem}

\subsection{Comparing Our Tail Optimality Condition to That of \texorpdfstring{\Citet{scully_characterizing_2020}}{Scully et al.~(2020c)}}
\label{rmk:prior_work_comparison}
Having formally stated our sufficient condition for tail optimality
in the heavy-tailed case,
we may now compare it in more detail to that of \citet{scully_characterizing_2020}.
Their condition \citep[Assumption~3.2]{scully_characterizing_2020}
and corresponding result \citep[Theorem~3.3]{scully_characterizing_2020}
are the same as our \cref{cond:exponents} and \cref{thm:heavy_soap}, respectively,
but restricted to the $\theta = 0$ case.
This means, roughly speaking, that \citet{scully_characterizing_2020}
only look at the lengths between ``peaks'' of the rank function,
without looking directly at the length of $w_x$-intervals.
For instance, the prior condition would treat the two rank functions shown in \cref{fig:zeta_vs_theta}
in the same way,
even though the $w_x$-intervals are much longer in \cref{fig:zeta_vs_theta}(a) than in \cref{fig:zeta_vs_theta}(b).

Unfortunately, looking only at distances between peaks of the rank function
is not enough to prove Gittins's tail optimality in the heavy-tailed case.
For Gittins, it turns out we cannot in general do better than setting $\eta = \infty$.
If we had to set $\theta = 0$,
then to make Gittins satisfy \cref{cond:exponents}, it turns out we would need to set $\zeta = 1$,
which is too large to satisfy \cref{eq:sufficient}.
But the Gittins rank function looks less like \cref{fig:zeta_vs_theta}(a)
and more like \cref{fig:zeta_vs_theta}(b):
even though the peaks can be far apart, there are ``gentle slopes'' between them.
Setting $\theta = 1$ and $\zeta = 0$ captures this behavior,
and this satisfies \cref{eq:sufficient}.
Thus, our refining of the sufficient condition of \citet[Assumption~3.2]{scully_characterizing_2020}
is necessary to prove Gittins's tail optimality in the heavy-tailed case,
at least for this SOAP-based approach.

Underlying the tail optimality results of \citet{scully_characterizing_2020}
is a busy period analysis combined with asymptotic response time bounds.
We make use of their busy period analysis,
which we distill into a simple statement
(\cref{thm:heavy_soap_intermediate}),
but we replace their asymptotic response time bounds with a sharper analysis
that accounts for the $\theta>0$ possibility
(\cref{sec:sufficient:upper, sec:sufficient:lower}).

Finally, we reiterate that while \citet{scully_characterizing_2020}
study only heavy-tailed size distributions,
we also study light-tailed job size distributions.
Our results for the light-tailed case are based on very different techniques,
reflecting fundamental differences between the heavy- and light-tailed settings.

\section{Light-Tailed Job Sizes}\label{sec:light}

Similarly to the previous section, we first define the class of light-tailed distributions and state the corresponding tail-optimality criterion in \cref{sec:preliminaries:light}. The main results in the light-tailed case, presented in \cref{sec:results:light}, are summarized as follows:
\begin{itemize}
    \item \Cref{thm:light_soap} classifies SOAP policies
    into tail-optimal, tail-intermediate, and tail-pessimal
    for light-tailed job sizes.
    \item \Cref{thm:light_gittins} shows that for light-tailed job sizes,
    Gittins can be any of tail-optimal, tail-intermediate, or tail-pessimal.
\end{itemize}
\begin{itemize}
    \item \Cref{thm:approximate_gittins} shows that making a small change to the Gittins rank function results in only a small change to mean response time.
    \item \Cref{thm:light_approximate_gittins} shows that for a wide class of light-tailed job size distributions
    for which Gittins is tail-pessimal,
    making a small change to Gittins's rank function results in a tail-optimal or -intermediate policy
    with mean response time arbitrarily close to Gittins's.
\end{itemize}

\subsection{Background on Light-Tailed Job Sizes}
\label{sec:preliminaries:light}

\begin{definition}
    \label{def:decay_rate}
    The \emph{decay rate} of random variable~$V$, denoted $d(V)$, is
    \begin{align*}
        d(V) = \lim_{t\to\infty} \frac{-\log \P{V > t}}{t}.
    \end{align*}
    That is, if the decay rate~$d(V)$ is finite, then $\P{V > t} = \exp(-d(V)t \pm o(t))$.
    Higher decay rates thus correspond to asymptotically lighter tails.
\end{definition}

Roughly speaking, the light-tailed job size distributions we study
are those with positive decay rate.
Our main tool for investigating the decay rate of a random variable~$V$
is via its Laplace-Stieltjes transform~$\lst{V}$, defined as
\begin{equation*}
    \lst{V}(s) = \E{\exp(-sV)} \in (0, \infty].
\end{equation*}
Under mild conditions on~$V$
\citep{mimica_exponential_2016, nakagawa_tail_2005, nakagawa_application_2007},
we can determine its decay rate in terms of the convergence of its Laplace-Stieltjes transform:
\begin{equation}
    \label{eq:decay_lst}
    d(V) = -\inf\curlgp{s \leq 0 \given \lst{V}(s) < \infty}.
\end{equation}

The specific class of light-tailed job size distributions we consider, described below,
are those which allow us to use \cref{eq:decay_lst} throughout this work
(\cref{sec:regularity}).
The class includes essentially all light-tailed distributions of practical interest
such as finite-support, phase-type, and Gaussian-tailed distributions.
In the terminology of \citet{abate_asymptotics_1997},
we consider all ``Class~I'' distributions.\footnote{%
    Our results can be generalized
    to some ``Class~II'' distributions \citep{abate_asymptotics_1997},
    which are also light-tailed.
    We comment on this in \cref{sec:regularity:class_ii}.
    However, working with Class~II distributions generally requires
    additional regularity or smoothness assumptions
    \citep[Section~5]{abate_asymptotics_1997},
    so for simplicity of presentation, we focus on Class~I distributions.}

\begin{definition}[Light-Tailed Job Size Distribution]
    \label{def:light}
    Given a job size distribution~$X$, let
    \begin{equation*}
        s^* = \inf\curlgp{s \leq 0 \given \lst{X}(s) < \infty}.
    \end{equation*}
    We say that $X$ is \emph{nicely light-tailed} if $s^* = -\infty$
    or $s^* \in (-\infty, 0)$ and $\lst{X}(s^*) = \infty$.
    In informal discussion, we omit ``nicely''.
\end{definition}

\begin{definition}[Tail Optimality in Light-Tailed Case]
    \label{def:light_optimality}
    Consider an M/G/1 with nicely light-tailed job size distribution~$X$.
    We say a scheduling policy~$\pi$ is
    \begin{itemize}
        \item \emph{log-tail-optimal} if $\pi$ maximizes $d(T_\pi)$,
        \item \emph{log-tail-pessimal} if $\pi$ minimizes $d(T_\pi)$, and
        \item \emph{log-tail-intermediate} otherwise.
    \end{itemize}
    In each case, we mean minimizing or maximizing over preemptive work-conserving policies.
    In informal discussion, we omit ``log-''.
\end{definition}

\subsection{Results for Light-Tailed Case}
\label{sec:results:light}

We have seen that a job's worst ever rank plays an important role in the heavy-tailed setting.
When the job sizes are light-tailed, we are interested in the age at which the rank function's \emph{global} maximum occurs.

\begin{definition}
    The \emph{worst age}, denoted~$a^*$,
    is the earliest age at which a job has the global maximum rank:
    \begin{align*}
        a^* = \inf\curlgp{a \in [0, \xmax) \given \forall b \in [0, \xmax), r(a) \geq r(b)}.
    \end{align*}
    If the rank function has no maximum, we define $a^* = \xmax$.
\end{definition}

As an example, FCFS has $a^* = 0$,
because a job's priority is the worst before it starts service.
In contrast, FB has $a^* = \xmax$,
because a job's priority gets strictly worse with age.

We already know that FCFS and FB are tail-optimal and tail-pessimal, respectively.
The theorem below fills in the gaps for all other SOAP policies,
showing that the performance of the response time tail is completely determined
by the worst age~$a^*$.

\begin{restatable:theorem}
    \label{thm:light_soap}
    Consider an M/G/1 with any nicely light-tailed job size distribution under a SOAP policy.
    Let $\xmax = \inf\curlgp{x \geq 0 \given \P{X > x} = 0}$.
    The policy is
    \begin{itemize}
        \item log-tail-optimal if $a^* = 0$,
        \item log-tail-intermediate if $0 < a^* < \xmax\esub$, and
        \item log-tail-pessimal if $a^* = \xmax\esub$.
    \end{itemize}
\end{restatable:theorem}

To apply \cref{thm:light_soap} to the Gittins policy,
we need to characterize how the job size distribution~$X$
affects Gittins's worst age~$a^*$.

\begin{definition}
    \label{def:nbue}
    We define two classes of distributions: $\NBUE$ and $\ENBUE$.
    \begin{itemize}
        \item We say $X$ is \emph{New Better than Used in Expectation}, writing $X \in \NBUE$, if for all ages $a \in [0, \xmax)$,
        \begin{equation*}
            \E{X} \geq \E{X - a \given X > a}.
        \end{equation*}
        \item We say $X$ is \emph{Eventually New Better than Used in Expectation}, writing $X \in \ENBUE$,
        if there exists $a_0 \in [0, \xmax)$ such $\gp{X - a_0 \given X > a_0} \in \NBUE$. Put another way, $X \in \ENBUE$ if there exists $a_0 \in [0, \xmax)$ such that for all $a \in [a_0, \xmax)$,
        \begin{equation*}
            \E{X - a_0 \given X > a_0} \geq \E{X - a \given X > a}.
        \end{equation*}
    \end{itemize}
\end{definition}

\begin{remark}
    It is well known that the $\NBUE$ class includes all distributions with (weakly) increasing hazard rate. Similarly, one can show the $\ENBUE$ class includes all distributions with ``eventually increasing'' hazard rate, meaning the hazard rate is increasing at all ages greater than some threshold $a_0 < \xmax$.
\end{remark}

Results of \citet{aalto_gittins_2009, aalto_properties_2011} connect the classes $\NBUE$ and $\ENBUE$
to Gittins's worst age~$a^*$, implying the following characterization.

\begin{restatable:theorem}
    \label{thm:light_gittins}
    Consider an M/G/1 with any nicely light-tailed job size distribution~$X$.
    Gittins is
    \begin{itemize}
        \item log-tail-optimal if $X \in \NBUE$,
        \item log-tail-intermediate if $X \in \ENBUE \setminus \NBUE$, and
        \item log-tail-pessimal if $X \not\in \ENBUE$.
    \end{itemize}
\end{restatable:theorem}

More generally, \cref{thm:light_soap} can imply an analogue of \cref{thm:light_gittins} for other SOAP policies whose rank at age~$a$ is related to the expected remaining size $\E{X - a \given X > a}$, such as the SERPT policy (\cref{rmk:serpt}).

The fact that Gittins can be log-tail-pessimal is intriguing, considering that it is optimal for mean response time, and tail-optimal under heavy-tailed job sizes.
Fortunately, in most cases where Gittins is log-tail-pessimal,
slightly tweaking Gittins yields a log-tail-intermediate policy
without sacrificing much mean response time performance.

\begin{definition}
    \label{def:approximate_gittins}
    A SOAP policy~$\pi$ is a \emph{$q$\=/approximate Gittins policy} if there exists a constant $m > 0$ such that
    for all ages $a \in [0, \xmax)$,
    \begin{equation*}
        \frac{r_\pi(a)}{r_\gittins(a)} \in [m, mq].
    \end{equation*}
    We may assume without loss of generality that $m = 1$,
    because the policy~$\pi'$ with rank function $r_{\pi'}(a) = r_\pi(a) / m$
    has identical behavior to policy~$\pi$.
\end{definition}

\begin{restatable:theorem}
    \label{thm:approximate_gittins}
    Consider an M/G/1 with any job size distribution.
    For any $q \geq 1$ and any $q$\=/approximate Gittins policy~$\pi$,\footnote{%
        This is a special case of a more general result \citep[Chapter~16]{scully_new_2022}, which appeared while this work was in revision.}
    \begin{equation*}
        \E{T_\pi} \leq q \E{T_\gittins}.
    \end{equation*}
\end{restatable:theorem}

An important observation is that a $q$\=/approximate Gittins policy has near-optimal mean response time for $q$ close to one. At the same time, changing the Gittins rank function even within a small factor $q$ can decrease the worst age, and therefore improve the tail performance.

\begin{restatable:theorem}
    \label{thm:light_approximate_gittins}
    Consider an M/G/1 with nicely light-tailed job size distribution $X \not\in \ENBUE$.
    Suppose that the expected remaining size of a job at all ages is uniformly bounded,
    meaning
    \begin{equation*}
        \sup_{a \in [0, \xmax)} \E{X - a \given X > a} < \infty.
    \end{equation*}
    Then for all $\epsilon > 0$,
    there exists a $(1 + \epsilon)$-approximate Gittins policy that is log-tail-optimal or log-tail-intermediate.
\end{restatable:theorem}

As an example in which log-tail-pessimality may be avoided, consider a hyperexponential job size $X$ with two different rates. That is, for $i=1,2$, with probability $p_i$ the job size is sampled from an exponential with rate $\mu_i$. Here, $p_1,p_2>0$ such that $p_1+p_2=1$ and we assume without loss of generality that $\mu_1>\mu_2$.

\begin{figure}
    \centering
    \input{figures/light_gittins_repair}
    \caption{The rank functions of the Gittins policy (translucent cyan curve) and a $q$-approximate Gittins policy (dotted yellow-green curve) for a hyperexponential distribution with rates $\mu_1$ and~$\mu_2$. By \cref{thm:light_soap}, these have different tail asymptotics. The Gittins rank function attains its supremum of $1/\mu_2$ in the $a \to \infty$ limit, so it is tail-pessimal. But the $q$-approximate Gittins rank function, described in \cref{eq:light_gittins_repair}, attains its supremum at a finite age~$\tilde{a}$, so it is tail-intermediate.}
    \label{fig:light_gittins_repair}
\end{figure}
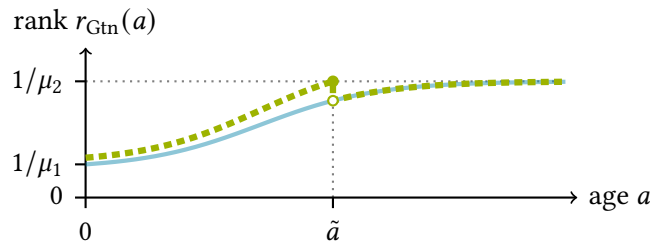

Because the hazard rate of a hyperexponential distribution is decreasing, $r_\gittins$ is increasing. Therefore, Gittins reduces to FB, so it is log-tail-pessimal. Fortunately though, $\E{X - a \given X > a} < 1/\mu_2$ for all ages $a$, hence \cref{thm:light_approximate_gittins} can be applied: for any $q>1$, the policy with rank function
\begin{equation}
    \label{eq:light_gittins_repair}
    r(a) = \begin{cases}
        qr_\gittins(a) \quad & \text{if } a \leq \tilde{a},\\
        r_\gittins(a) \quad & \text{otherwise},
    \end{cases}
\end{equation}
where $\tilde{a}$ is such that $qr_\gittins(\tilde{a}) = 1/\mu_2$, is a $q$-approximate Gittins policy with $a^*=\tilde{a}<\infty$. See \cref{fig:light_gittins_repair}.

While \cref{thm:light_approximate_gittins} implies any approximation factor $q > 1$ suffices to prevent log-tail-pessimality, there is a tradeoff to be made when choosing~$q$. Higher values of $q$ result in worse guarantees for the mean (\cref{thm:approximate_gittins}). But higher values of $q$ lead to lower values of~$a^*$, which in turn result in better tail decay (\cref{lem:light_soap_decay}). We leave exploring this tradeoff quantitatively to future work.

\subsection{Proof Organization}

The remainder of this paper is organized as follows. Necessary background and notation on the SOAP framework is given in \cref{sec:soap}. Then, we prove our results for heavy tails, \cref{thm:heavy_soap, thm:heavy_gittins}, respectively in \cref{sec:heavy_soap, sec:heavy_gittins}. Similarly, proofs of our results for the light-tailed case are given in Sections~\ref{sec:light_soap} (\cref{thm:light_soap}) and \ref{sec:light_gittins} (Theorems \ref{thm:light_gittins} and \ref{thm:light_approximate_gittins}). The remaining main result, \cref{thm:approximate_gittins}, requires substantially more technical machinery for its proof, which is why we defer it to \cref{sec:gittins}. Finally, \cref{sec:conclusion} describes how our results answer the questions posed in \cref{sec:intro}.

\section{SOAP Notation}\label{sec:soap}

We use the following notations related to SOAP policies,
which are standard in the literature
\citep{scully_soap_2018, scully_characterizing_2020, scully_simple_2020}.
These definitions are necessary for writing down and working with
the response time formulas of SOAP policies.
All of these definitions are given in terms of a SOAP policy with rank function~$r$. Throughout, it will always be clear from context which rank function is being referred to.

\begin{definition}
    \label{def:kth_relevant_age_interval}
    The \emph{$k$th maximal $w$-interval} is $(b_k[w], c_k[w])$, where for all $k \geq 1$\footnote{%
        In the infimum expressions below,
        all sets of ages are implicitly assumed to be subsets of $[0, \xmax)$,
        and the infimum of an empty set is taken to be~$\xmax$.}
    \begin{align*}
        b_0[w] &= 0,
        \iftwocoleq{
            \\
            c_0[w] &= \inf\curlgp{a \geq 0 \given r(a) > w}, \\
            b_k[w] &= \inf\curlgp{a > c_{k - 1}[w] \given r(a) \leq w}, \\
        }{
            &
            b_k[w] &= \inf\curlgp{a > c_{k - 1}[w] \given r(a) \leq w}, \\
            c_0[w] &= \inf\curlgp{a \geq 0 \given r(a) > w}, &
        }
        c_k[w] &= \inf\curlgp{a > b_k[w] \given r(a) > w}.
    \end{align*}
    Additionally, let $\kmax[w]$ be the maximum~$k$ such that $b_k[w] < \xmax$.
    It may be that $\kmax[w] = \infty$.
\end{definition}

One may easily check that $(b_k[w], c_k[w])$ is indeed a $w$-interval, with only one exception:
it may be that $b_0[w] = c_0[w] = 0$,
in which case the interval is empty and thus not a $w$-interval.
See \cref{fig:w_x-intervals}, whose orange regions are specifically the zeroth, first, and second maximal $w_x$-intervals for the pictured size~$x$.

\begin{definition}
    \label{def:kth_relevant_segment}
    The \emph{$k$th $w$-relevant job segment} is the random variable
    \begin{equation*}
        X_k[w] = \max\{0, \min\{X, c_k[w]\} - b_k[w]\}.
    \end{equation*}
    For convenience, we define $X_k[w] = 0$ for $k > \kmax[w]$.
\end{definition}

The following lemma gives a convenient formula for moments of relevant job segments.

\begin{lemma}[\textnormal{\citet[Lemma~6.16]{scully_characterizing_2020}}]
    \label{lem:kth_relevant_moment}
    \begin{equation*}
        \E{X_k[w]^{p + 1}}
        = \int_{b_k[w]}^{c_k[w]} (p + 1) (t - b_k[w])^p \tail{t} \d{t}.
    \end{equation*}
\end{lemma}

\begin{definition}
    \label{def:yz}
    For a job of size~$x$, we define\footnote{%
        Here and throughout, the postscript $-$ denotes the limit from the left of a right-continuous function with left limits. In this case, $c_0[w_x-] = \lim_{\epsilon \downarrow 0} c_0[w_x - \epsilon]$.}
    \begin{align*}
        y_x &= c_0[w_x-], \iftwocoleq{\\}{&}
        z_x &= c_0[w_x].
    \end{align*}
    See \cref{fig:yz}. Intuitively, $y_x$ is the earliest age at which a job of size~$x$ attains its worst ever rank~$w_x$, and $z_x$ is the earliest age at which the rank function exceeds~$w_x$.
\end{definition}

\begin{figure}
    \centering
    \input{figures/yz}
    \caption{%
        Illustration of $y_x$ and $z_x$ (\cref{def:yz})}
    \label{fig:yz}
\end{figure}

Note that it may be that $y_x = x = z_x$. This occurs if the rank function is strictly increasing at~$x$, and $x$ is a ``running maximum'', meaning $r(a) < r(x)$ for all ages $a \in [0, x)$.

Our final piece of notation is a generalization of a job's worst ever rank.

\begin{definition}
    \label{def:worst_future_rank}
    The \emph{worst future rank} of a job of size~$x$ at age~$a$, written $w_x(a)$, is
    \begin{equation*}
        w_x(a) = \sup_{a \leq b < x} r(b).
    \end{equation*}
    Note that the worst ever rank is simply the worst future rank at age~$0$, i.e. $w_x = w_x(0)$.
\end{definition}

One can write a formula characterizing~$T(x)$, the response time distribution of jobs of size~$x$, in terms of the notation from \cref{def:kth_relevant_segment, def:worst_future_rank}. We refer the reader to \citet{scully_soap_2018} for details, noting here one simple consequence of that formula.

\begin{lemma}
    \label{lem:stochastically_increasing}
    Under any SOAP policy, the response time of jobs of size~$x$ is stochastically increasing in~$x$. That is, for all $x_1 \geq x_0 > 0$, we have $T(x_0) \leq_\st T(x_1)$.
\end{lemma}

\section{Heavy-Tailed Job Sizes: Tail Asymptotics of SOAP Policies}
\label{sec:heavy_soap}

This section is devoted to proving \cref{thm:heavy_soap}, which gives a sufficient condition for a SOAP policy to be tail-optimal. Our first step is to invoke a result of \citet{scully_characterizing_2020}.\footnote{%
    While \citet{scully_characterizing_2020} do not explicitly state this result, namely our \cref{thm:heavy_soap_intermediate}, they prove it as an intermediate step towards another result \citep[Section~4]{scully_characterizing_2020}.}
Their result reduces the task of proving tail optimality to a much simpler task, namely bounding various functions of moments of $w_x$-relevant job segments. To state their result, we use a ``polynomially strict'' version of little-$o$ notation, which we write as~$\ostrict(\cdot)$.

\begin{definition}
    For $p > 0$, the notation $\ostrict(x^p)$ stands for $O(x^{p - \epsilon})$ for some unspecified $\epsilon > 0$.
\end{definition}

\begin{condition}
    \label{cond:X_upper}
    There exists $q > \beta$ such that for all $p \in (0, q)$,
    \begin{equation}
        \label{eq:X_upper_goal}
        \sum_{k = 0}^{\kmax[w_x]} \E{X_k[w_x]^{p + 1}} \leq \ostrict(x^p).
    \end{equation}
\end{condition}

\begin{condition}
    \label{cond:R_lower}
    \begin{equation*}
        \int_0^x \frac{1}{1 - \lambda \E{X_0[w_x(a)-]}} \d{a} \geq \frac{x}{1 - \rho} - \ostrict(x).
    \end{equation*}
\end{condition}

\begin{lemma}[\textnormal{\citet{scully_characterizing_2020}}]
    \label{thm:heavy_soap_intermediate}
    Consider an M/G/1 with nicely heavy-tailed job size under a SOAP policy.
    \Cref{cond:R_lower, cond:X_upper} together imply tail optimality.
\end{lemma}

Our proof of \cref{thm:heavy_soap} is based thus based on verifying \cref{cond:X_upper, cond:R_lower}.

\begin{proof}[Proof of \cref{thm:heavy_soap}]
    By \cref{thm:heavy_soap_intermediate}, it suffices to show that \cref{cond:exponents, eq:sufficient} together imply \cref{cond:X_upper, cond:R_lower}. We do so in \cref{lem:X_upper} (see \cref{sec:sufficient:upper}), which handles \cref{cond:X_upper}, and \cref{lem:R_lower} (see \cref{sec:sufficient:lower}), which handles \cref{cond:R_lower}.
\end{proof}

\subsection{Showing \texorpdfstring{\Cref{cond:X_upper}}{Condition~7.2}}
\label{sec:sufficient:upper}

Our goal in this section is to prove the following.

\begin{proposition}
    \label{lem:X_upper}
    If \cref{eq:sufficient} holds, then \cref{cond:exponents} implies \cref{cond:X_upper}.
\end{proposition}

That is, we want to show \cref{eq:X_upper_goal} under certain conditions. The first step is to compute the left-hand side of \cref{eq:X_upper_goal} while assuming only \cref{cond:exponents}. The following lemma does so, separating out the $k = 0$ term because it has a slightly different form.

\begin{restatable:lemma}
    \label{lem:X_upper_computation}
    Suppose \cref{cond:exponents} holds.
    \begin{enumerate:lemma}
        \item For all $p \geq 0$,
        \begin{align*}
            \iftwocoleq{\MoveEqLeft}{}
            \E{X_0[w_x]^{p + 1}}
            \iftwocoleq{\\* &}{}
            \leq \begin{cases}
                O(1) & \text{if } p < \alpha - 1 \\
                O(\log x) & \text{if } p = \alpha - 1 \\
                O(x^{\max\{1, \zeta + \theta\}(p - \alpha + 1)}) & \text{if } p > \alpha - 1.
            \end{cases}
        \end{align*}
        \item For all $p \geq 0$,
        \begin{align*}
            \iftwocoleq{\MoveEqLeft}{}
            \sum_{k = 1}^{\kmax[w_x]} \E{X_k[w_x]^{p + 1}}
            \iftwocoleq{\\* &}{}
            \leq \begin{cases}
                O(x^{\theta p + \zeta p - \alpha + 1}) & \text{if } \zeta p < \alpha - 1 \\
                O(x^{\theta p} \log x^\eta) & \text{if } \zeta p = \alpha - 1 \\
                O(x^{\theta p + \eta(\zeta p - \alpha + 1)}) & \text{if } \zeta p > \alpha - 1.
            \end{cases}
        \end{align*}
    \end{enumerate:lemma}
\end{restatable:lemma}

The proof of \cref{lem:X_upper_computation} is largely computational, so we defer it to \cref{sec:computations:heavy}.

\begin{proof}[Proof of~\cref{lem:X_upper}]
    Our goal is to choose $q > \beta$ such that for all $p \in (0, q)$, \cref{eq:X_upper_goal} holds. Specifically, we choose
    \begin{equation*}
        q = \begin{dcases}
            \frac{\alpha - 1}{d} & \text{if } d > 0 \\
            \infty & \text{if } d \leq 0,
        \end{dcases}
    \end{equation*}
    where
    \begin{equation*}\begin{aligned}
        d
        &= \text{left-hand side of \cref{eq:sufficient}}
        \iftwocoleq{\\ &}
        = \zeta + (\theta - 1)^+ - \frac{(1 - \theta)^+}{\eta}.
    \end{aligned}\end{equation*}
    We have $q > \beta$ by \cref{eq:sufficient}, and an analogue of \cref{eq:sufficient} holds for all $p \in (0, q)$, namely
    \begin{equation}
        \label{eq:sufficient_p}
        \zeta + (\theta - 1)^+ - \frac{(1 - \theta)^+}{\eta}
        < \frac{\alpha - 1}{p}.
    \end{equation}

    By \cref{lem:X_upper_computation}, it suffices to show that for all $p \in (0, q)$, both of the following hold:
    \begin{align*}
        p &> \begin{cases}
            0 & \text{if } p < \alpha - 1 \\
            0 & \text{if } p = \alpha - 1 \\
            \max\{1, \zeta + \theta\}(p - \alpha + 1) & \text{if } p > \alpha - 1,
        \end{cases}
        \eqsidetext{see \cref{lem:X_upper_computation}(i)}
        \\[\smallskipamount]
        p &> \begin{cases}
            \theta p + \zeta p - \alpha + 1 & \text{if } \zeta p < \alpha - 1 \\
            \theta p & \text{if } \zeta p = \alpha - 1 \\
            \theta p + \eta(\zeta p - \alpha + 1) & \text{if } \zeta p > \alpha - 1.
        \end{cases}
        \eqsidetext{see \cref{lem:X_upper_computation}(ii)}
    \end{align*}
    Under the assumptions of \cref{def:heavy, cond:exponents}, namely
    \begin{equation}
        \label{eq:exponent_assumptions}
        \alpha > 1, \quad \zeta > 0, \quad \theta > 0, \quad \text{and} \quad \eta \geq \max\{1, \zeta + \theta\},
    \end{equation}
    this is equivalent to showing both of
    \begin{align}
        \label{eq:X0_upper_goal_p}
        \max\{1, \zeta + \theta\}\gp*{1 - \frac{\alpha - 1}{p}}^{\!\!+}\! &< 1,
        \\
        \label{eq:Xk_upper_goal_p}
        \theta - \gp*{\frac{\alpha - 1}{p} - \zeta}^{\!\!+}\! + \eta \gp*{\zeta - \frac{\alpha - 1}{p}}^{\!\!+}\! &< 1.
    \end{align}
    We show below that both of these are implied by \cref{eq:sufficient_p}. To reduce clutter, let
    \begin{equation*}
        \nu = \frac{\alpha - 1}{p}.
    \end{equation*}
    For \cref{eq:X0_upper_goal_p}, we compute
    {\allowdisplaybreaks
    \begin{align*}
        \iftwocoleq{&}{}
        \text{\cref{eq:X0_upper_goal_p}}
        \iftwocoleq{\\*}{\ }
        \Leftrightarrow\ \ &(1 - \nu)^+ < \frac{1}{\max\{1, \zeta + \theta\}} \\*
        \Leftrightarrow\ \ &[\zeta + \theta \leq 1] \ \lor\ \gp[\bigg]{[\zeta + \theta > 1] \ \land\ \sqgp[\bigg]{\frac{\zeta + \theta - 1}{\zeta + \theta} < \nu}}
            \iftwocoleq{\eqsidetext[{\\&}]}{\eqsidetext[{&&\mkern-48mu}]}{by \cref{eq:exponent_assumptions} $\Rightarrow$ $\nu > 0$} \\
        \Leftrightarrow\ \ &
            [\zeta + \theta \leq 1] \ \lor\ \biggl(
                [\zeta + \theta > 1] \iftwocoleq{\\&\quad}{\ } \land\ \sqgp[\bigg]{
                    \frac{\zeta + (\theta - 1)^+ - (1 - \theta)^+}{\zeta + \theta} < \nu
                }
            \biggr) \\
        \Leftarrow\ \ &
            [\zeta + \theta \leq 1] \ \lor\ \biggl(
                [\zeta + \theta > 1] \iftwocoleq{\\&\quad}{\ } \land\ \sqgp[\bigg]{
                    \zeta + (\theta - 1)^+ - \frac{(1 - \theta)^+}{\eta} < \nu
                }
            \biggr)
            \iftwocoleq{\eqsidetext[{\\&}]}{\\&\eqsidetext[{&&\mkern-48mu}]}{by \cref{eq:exponent_assumptions} $\Rightarrow$ $\eta \geq \zeta + \theta$} \\*
        \Leftarrow\ \ &\text{\cref{eq:sufficient_p}},
    \end{align*}
    and for \cref{eq:Xk_upper_goal_p}, we compute
    \begin{align*}
        \iftwocoleq{&}{}
        \text{\cref{eq:Xk_upper_goal_p}}
        \iftwocoleq{\\*}{\ }
        \Leftrightarrow\ \ &
            \gp[\Big]{
                [\zeta < \nu]
                \ \land\ [\theta - \nu + \zeta < 1]
            } \iftwocoleq{\\&{}}{\ } \lor\ \gp[\Big]{
                [\zeta \geq \nu]
                \ \land\ [\theta - \eta\nu + \eta\zeta < 1]
            } \\
        \Leftrightarrow\ \ &
            \gp[\Big]{
                [\zeta < \nu]
                \ \land\ [\zeta + \theta - 1 < \nu]
            } \iftwocoleq{\\&{}}{\ } \lor\ \gp[\bigg]{
                [\zeta \geq \nu]
                \ \land\ \sqgp[\bigg]{\zeta - \frac{1 - \theta}{\eta} < \nu}
            }
            \iftwocoleq{\eqsidetext[{\\&}]}{\\&\eqsidetext[{&&\mkern-48mu}]}{by \cref{eq:exponent_assumptions} $\Rightarrow$ $\eta > 0$} \\
        \Leftrightarrow\ \
&            \gp[\Big]{
                [\theta \geq 1]
                \ \land\ [\zeta + \theta - 1 < \nu]
            } \iftwocoleq{\\&{}}{\ } \lor\ \gp[\bigg]{
                [\theta < 1]
                \ \land\ \sqgp[\bigg]{\zeta - \frac{1 - \theta}{\eta} < \nu}
            } \\
        \Leftrightarrow\ \ &\text{\cref{eq:sufficient_p}}.
        & & \qedhere
    \end{align*}}
\end{proof}

\begin{remark}
    Note that in addition to \cref{eq:sufficient_p} implying \cref{eq:Xk_upper_goal_p}, the reverse implication also holds. This suggests that the precondition of \cref{thm:heavy_soap}, and in particular \cref{eq:sufficient}, cannot be easily relaxed.
\end{remark}

\subsection{Showing \texorpdfstring{\Cref{cond:R_lower}}{Condition~7.3}}
\label{sec:sufficient:lower}

Our goal in this section is to prove \cref{lem:R_lower} below. It is applicable to proving \cref{thm:heavy_soap} because \cref{eq:sufficient} implies its precondition.

\begin{proposition}
    \label{lem:R_lower}
    If $\zeta < 1$ or $\eta < \infty$, then \cref{cond:exponents} implies \cref{cond:R_lower}.
\end{proposition}

\begin{proof}
    The $\eta < \infty$ case follows from a result of \citet[Lemma~7.3]{scully_characterizing_2020},
    so we address only the $\zeta < 1$ case. We first observe that for all $\rho' \in [0, \rho]$, we have $\frac{1}{1 - \rho'} \geq \frac{1}{1 - \rho} - \frac{\rho - \rho'}{(1 - \rho)^2}$. This means \cref{cond:R_lower} holds if
    \begin{equation}
        \label{eq:R_lower_goal}
        \int_0^x (\E{X} - \E{X_0[w_x(a)-]}) \d{a} \leq \ostrict(x).
    \end{equation}
    We can rewrite the integrand as
    \begin{align*}
        \iftwocoleq{\MoveEqLeft}{}
        \E{X} - \E{X_0[w_x(a)-]}
        \iftwocoleq{\\*}{}
        &= \int_0^\infty \tail{t} \d{t} - \int_0^{c_0[w_x(a)-]} \tail{t} \d{t}
        \byref{lem:kth_relevant_moment}
        \\
        &\leq \int_{c_0[w_x(a)-]}^\infty O(t^{-\alpha}) \d{t}
        \byref{def:heavy}
        \\
        &\leq O(c_0[w_x(a)-]^{-(\alpha - 1)}).
    \end{align*}
    Of course, the integrand is also bounded above by $\E{X}$, so
    \begin{align*}
        \iftwocoleq{\MoveEqLeft}{}
        \int_0^x (\E{X} - \E{X_0[w_x(a)-]}) \d{a}
        \iftwocoleq{\\*}{}
        \yestag
        \label{eq:R_lower_integral}
        &\leq \int_0^x O(\min\{1, c_0[w_x(a)-]^{-(\alpha - 1)}\}) \d{a}.
    \end{align*}

    A job of size~$x$ attains its worst ever rank~$w_x$ at age~$y_x$.
    This means that for all $a < y_x$,
    we have $w_x(a) = w_x$, which by \cref{def:yz} implies $c_0[w_x(a)-] = y_x$.
    Splitting the integral in \cref{eq:R_lower_integral} at $a = y_x$ yields
    \begin{align*}
        \iftwocoleq{
            \mathrlap{\int_0^x (\E{X} - \E{X_0[w_x(a)-]}) \d{a}} \quad& \\*
        }{
            \MoveEqLeft \int_0^x (\E{X} - \E{X_0[w_x(a)-]}) \d{a} \\*
        }
        &\leq \int_0^{y_x} O(y_x^{-(\alpha - 1)}) \d{a}
            \iftwocoleq{\\* &\quad}{}
            + \int_{y_x}^x O(\min\{1, c_0[w_x(a)-]^{-(\alpha - 1)}\}) \d{a} \\
        \yestag
        \label{eq:R_lower_progress}
        &\leq y_x^{(2 - \alpha)^+} + \int_{y_x}^x O(\min\{1, c_0[w_x(a)-]^{-(\alpha - 1)}\}) \d{a}.
    \end{align*}
    Because $y_x \leq x$ and $\alpha > 1$, it suffices to show
    the integral in \cref{eq:R_lower_progress} is $\ostrict(x)$.

    The main remaining obstacle is bounding $c_0[w_x(a)-]$.
    We do so in \cref{lem:c0_bound} below,
    which states
    \begin{equation*}
        c_0[w_x(a)-] \geq \Omega\gp*{\gp*{\frac{x - a}{x^\zeta}}^{1/\kappa}}
    \end{equation*}
    for some $\kappa \geq 2(\alpha - 1)$.
    Plugging this into \cref{eq:R_lower_progress}
    and substituting $u = x^{-\zeta}(x - a)$ gives
    \begin{align*}
        \iftwocoleq{\MoveEqLeft}{}
        \int_{y_x}^x O(\min\{1, c_0[w_x(a)-]^{-(\alpha - 1)}\}) \d{a}
        \iftwocoleq{\\*}{}
        &\leq \int_0^x O\gp*{\min\curlgp*{1, \gp*{\frac{x - a}{x^\zeta}}^{-(\alpha - 1)/\kappa}}} \d{a} \\
        &\leq \int_{x - x^\zeta}^x O(1) \d{a} + x^\zeta \int_1^{x^{1 - \zeta}} O(u^{-(\alpha - 1)/\kappa}) \d{u} \\
        &= O(x^\zeta) + O(x^{\zeta + (1 - \zeta)\gp{1 - (\alpha - 1)/\kappa}}).
    \end{align*}
    Because $\zeta < 1$ and $(\alpha - 1)/\kappa \in (0, 1/2]$, this is~$\ostrict(x)$, as desired.
\end{proof}

The following lemma bounds $c_0[w_x(a)-]$. We defer its proof to \cref{sec:computations:heavy}.

\begin{restatable:lemma}
    \label{lem:c0_bound}
    Suppose \cref{cond:exponents} holds,
    and let $\kappa = 2\max\{\alpha - 1, \theta\}$.
    For all $x \geq 0$ and $a \in (y_x, x)$,
    \begin{equation*}
        c_0[w_x(a)-] \geq \Omega\gp*{\gp*{\frac{x - a}{x^\zeta}}^{1/\kappa}}.
    \end{equation*}
\end{restatable:lemma}

\section{Heavy-Tailed Job Sizes: Gittins is Tail-Optimal}
\label{sec:heavy_gittins}

In this section we prove \cref{thm:heavy_gittins}, namely that Gittins is tail-optimal for heavy-tailed job sizes. Specifically, we will show that Gittins satisfies \cref{cond:exponents} with
\begin{align*}
    \zeta = 0, \quad
    \theta = 1, \quad
    \text{and} \quad
    \eta = \infty.
\end{align*}
This suffices because with the above values, \cref{eq:sufficient} holds for all $\beta \geq \alpha > 1$, so \cref{thm:heavy_soap} implies tail optimality. In fact, \cref{cond:exponents_max} holds trivially when $\eta = \infty$, so only \cref{cond:exponents_length} remains.

Our goal is thus show that for the Gittins rank function,\footnote{%
    Recall from \cref{def:relevant_age_interval}(ii) that a $w$-interval $(b, c)$ is right-maximal if, roughly speaking, $c$ is as large as possible for a $w$-interval starting at~$b$.}
\begin{equation}
    \label{eq:heavy_gittins_goal}
    \iftwocoleq{\parbox{14.5em}}{\text}{any right-maximal $w_x$-interval $(b, c)$ with $b \geq x$ has length $c - b \leq O(x)$.}
\end{equation}
In words, \cref{eq:heavy_gittins_goal} says that whenever the Gittins rank function dips below the worst rank a job of size~$x$ ever has, it does so for an interval of length at most~$O(x)$. See \cref{fig:heavy_gittins_goal}.

\begin{figure}
    \centering
    \input{figures/heavy_gittins_goal}
    \caption{%
        Illustration of \cref{eq:heavy_gittins_goal}: any $w_x$-interval (orange regions), namely any interval where the Gittins rank function (cyan curve) is better than the worst ever rank~$w_x$ of a job of size~$x$, has length $O(x)$.}
    \label{fig:heavy_gittins_goal}
\end{figure}
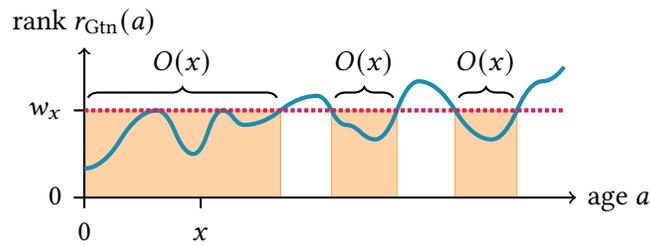

\subsection{The Time-Per-Completion Function}

To make further progress, we need to consider the specific form of the Gittins rank function. One useful way of thinking about the Gittins rank function uses the following definition \citep{aalto_gittins_2009, aalto_properties_2011}.

\begin{definition}
    \label{def:time-per-completion}
    The \emph{time-per-completion} function for a given job size distribution~$S$ is defined for $0 \leq b < c \leq \xmax$ as
    \begin{equation*}
        \phi(b, c)
        = \frac{\E{\min\{S, c\} - b \given S > b}}{\P{S \leq c \given S > b}}
        = \frac{\int_b^c \tail{t} \d{t}}{\tail{b} - \tail{c}}.
    \end{equation*}
    The intuition is that if we serve a job from age~$b$ until either age~$c$ or its completion, whichever comes first, then we can interpret $\phi(b, c)$ as
    \begin{equation*}
        \phi(b, c) = \frac{\E*{\iftwocoleq{\parbox{9em}}{\text}{time spent serving job from age~$b$ to age~$c$}}}{\E*{\1\gp*{\iftwocoleq{\parbox{11em}}{\text}{job completes when served from age~$b$ to age~$c$}}}},
    \end{equation*}
    which is an expected amount of time over an expected number of completions.
\end{definition}

One can rewrite the Gittins rank function (\cref{def:gittins}) in terms of $\phi(b, c)$ as
\begin{equation}
    \label{eq:gittins_time-per-completion}
    r_\gittins(a) = \inf_{c > a} \phi(a, c).
\end{equation}
That is, under Gittins, a job's rank at age~$a$ is the best possible time-per-completion ratio achievable on any interval starting at age~$a$.

How does using the time-per-completion function help us show~\cref{eq:heavy_gittins_goal}? The key observation is that for the Gittins rank function to be low over an interval~$(b, c)$, a job must be relatively likely to complete during $(b, c)$, which would imply $\phi(b, c)$ is also low. The following lemma, which we prove in \cref{pf:phi_w-relevant}, formalizes this intuition.

\begin{restatable:lemma}
    \label{lem:phi_w-relevant}
    Under Gittins, for any right-maximal $w$-interval $(b, c)$, we have $\phi(b, c) \leq w$.
\end{restatable:lemma}

We note that while many results similar to \cref{lem:phi_w-relevant} have been shown in prior work \citep{aalto_gittins_2009, aalto_properties_2011, scully_characterizing_2020}, to the best of our knowledge, \cref{lem:phi_w-relevant} itself is new.

With \cref{lem:phi_w-relevant} in hand, showing \cref{eq:heavy_gittins_goal} amounts to
\begin{itemize}
\item
    proving an upper bound on~$w_x$, and
\item
    proving a lower bound on $\phi(b, c)$ for all right-maximal $w_x$-intervals $(b, c)$.
\end{itemize}
The first of these follows simply from prior work: \citet[Section~3.2]{scully_characterizing_2020} show that the Gittins rank function is bounded by $r_\gittins(a) \leq O(a)$, which implies
\begin{equation}
    \label{eq:w_bound}
    w_x = \sup_{0 \leq a < x} r_\gittins(a) \leq \sup_{0 \leq a < x} O(a) \leq O(x).
\end{equation}
It thus remains only to bound $\phi(b, c)$ below. We begin with a bound on $\phi(b, c)$ from prior work.

\begin{lemma}[\textnormal{\citet[Lemma~6.8]{scully_optimal_2021}}]
    \label{lem:phi_bound}
    For any nicely heavy-tailed job size distribution and all $c > b \geq 0$, the time-per-completion function is bounded by
    \begin{equation*}
        \phi(b, c) \geq \Omega\gp*{\frac{b}{c}(c - b)}.
    \end{equation*}
\end{lemma}

Combining \cref{eq:w_bound, lem:phi_bound, lem:phi_w-relevant}, we find that for any right-maximal $w_x$-interval $(b, c)$,
\begin{equation*}\begin{aligned}
    c - b
    &\leq \frac{c}{b} \, O(\phi(b, c))
        \byref[\eqsidetextsepsameline]{lem:phi_bound} \\
    &\leq \frac{c}{b} \, O(w_x)
        \byref[\eqsidetextsepsameline]{lem:phi_w-relevant} \\
    &\leq \gp*{1 + \frac{c - b}{b}} \, O(x).
        \byref[\eqsidetextsepsameline]{eq:w_bound}
\end{aligned}\end{equation*}
Therefore, to show \cref{eq:heavy_gittins_goal} and thereby \cref{thm:heavy_gittins}, it suffices to show that
\begin{equation}
    \label{eq:heavy_gittins_goal_progress}
    \iftwocoleq{\parbox{14.5em}}{\text}{any right-maximal $w_x$-interval $(b, c)$ with $b \geq x$ has length $c - b \leq O(b)$.}
\end{equation}
We could equivalently write $c \leq O(b)$, but the $c - b \leq O(b)$ form emphasizes the progress we have made relative to \cref{eq:heavy_gittins_goal}: we have weakened our goal from proving an $O(x)$ bound to an $O(b)$ bound.

\subsection{Bounding Lengths of \texorpdfstring{$w_x$-Intervals}{wx-Intervals}}

There is one more fact from prior work that we need to prove~\cref{eq:heavy_gittins_goal_progress}.

\begin{lemma}[\textnormal{\citet[Theorem~6.4]{scully_optimal_2021}}]
    \label{lem:yz_bound}
    Under Gittins, $y_x = \Theta(x)$ and $z_x = \Theta(x)$.
\end{lemma}

In general, $(y_x, z_x)$ is a right-maximal $w_x$-interval, but there can be plenty of other $w_x$-intervals starting at values greater than~$z_x$. How can we use \cref{lem:yz_bound} to study such intervals? The key is to look not at $y_x$ and~$z_x$, but at $y_u$ and~$z_u$, where $u$ is some point inside the $w_x$-interval whose length we wish to bound.

Specifically, consider a right-maximal $w_x$-interval $(b, c)$ with $b \geq x$, and let $u \in (b, c)$ be an arbitrary size in the interval. \Cref{fig:ybcz} illustrates the relationship between the interval $(b, c)$, the worst ever rank~$w_u$ of size~$u$, and the ages $y_u$ and~$z_u$ (\cref{def:yz}). The figure suggests the following lemma, which is a slight generalization of a result of \citet[Lemma~6.18]{scully_characterizing_2020}.

\begin{figure}
    \centering
    \input{figures/ybcz}
    \caption{%
        A $w_x$-interval $(b, c)$ contained inside $(y_u, z_u)$, where $u \in (b, c)$ is an arbitrary point in the $w_x$-interval.}
    \label{fig:ybcz}
\end{figure}
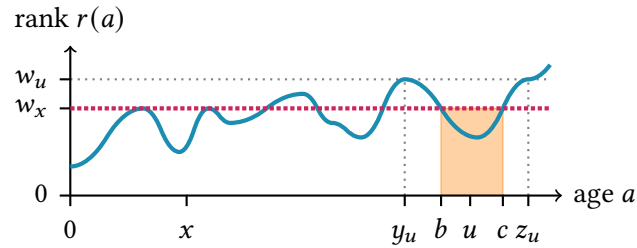

\begin{lemma}
    \label{lem:ybcz}
    Under any SOAP policy, for any $w_x\esub$-interval $(b, c)$ with $b \geq x$ and any $u \in (b, c)$,
    \begin{equation*}
        y_u \leq b < c \leq z_u.
    \end{equation*}
\end{lemma}

\begin{proof}
    It is clear from \cref{def:yz} that $y_u \leq u \leq z_u$. It thus suffices to show that $y_u \not\in (b, c)$ and $z_u \not\in (b, c)$. Both steps use the fact that $x \leq b < u$ implies $w_x \leq w_u$ (\cref{def:worst_ever_rank}).

    We first show that $z_u \not\in (b, c)$. By \cref{def:relevant_age_interval}, the rank function does not exceed~$w_x$ over the interval $(b, c)$. But \cref{def:yz} implies that every neighborhood of~$z_u$ contains a point whose rank is greater than~$w_u$. Because $w_u \geq w_x$, it must be that $z_u \not\in (b, c)$.

    If $w_u > w_x$, then the argument that $y_u \not\in (b, c)$ is analogous to that for~$z_u$. The difference is that this time, \cref{def:yz} implies that for all $\epsilon > 0$, every neighborhood of~$y_u$ contains a point whose rank is greater than~$w_u - \epsilon$. Because $w_u - \epsilon > w_x$ for small enough~$\epsilon$, it must be that $y_u \not\in (b, c)$.

    If instead $w_u = w_x$, then $y_u = y_x$, and we can reason more simply: \cref{def:yz} implies $y_x \leq x$, and we have assumed $x \leq b$, so $y_u \leq b$.
\end{proof}

\Cref{lem:yz_bound, lem:ybcz} combine to prove \cref{eq:heavy_gittins_goal_progress}, from which \cref{thm:heavy_gittins} follows.

\begin{proof}[Proof of \cref{thm:heavy_gittins}]
    Let $(b, c)$ be a right-maximal $w_x$-interval with $b \geq x$, and let $u \in (b, c)$. We compute
    \begin{equation*}\begin{aligned}
        c - b
        &\leq \frac{z_u}{y_u} \cdot \frac{b}{c}(c - b)
        \byref[\eqsidetextsepsameline]{lem:ybcz} \\
        &\leq O(1) \cdot \frac{b}{c}(c - b)
        \byref[\eqsidetextsepsameline]{lem:yz_bound} \\
        &\leq O(\phi(b, c))
        \byref[\eqsidetextsepsameline]{lem:phi_bound} \\
        &\leq O(w_x)
        \byref[\eqsidetextsepsameline]{lem:phi_w-relevant} \\
        &\leq O(x).
        \byref[\eqsidetextsepsameline]{eq:w_bound}
    \end{aligned}\end{equation*}
    The fact that $c - b \leq O(x)$ means that Gittins satisfies \cref{cond:exponents} with $\zeta = 0$, $\theta = 1$, and $\eta = \infty$. These obey \cref{eq:sufficient}, so by \cref{thm:heavy_soap}, Gittins is tail-optimal.
\end{proof}

\section{Light-Tailed Job Sizes: Tail Asymptotics of SOAP Policies}
\label{sec:light_soap}

In this section we prove our main theorem for light-tailed job sizes.

\restate*{thm:light_soap}
\begin{proof}
The result follows from \cref{lem:light-optimal,lem:light-pessimal,lem:light-intermediate},
which we prove in the rest of this section.
\end{proof}

\subsection{Tail-Optimal Case}
\label{sec:light_soap:optimal}

\begin{proposition}\label{lem:light-optimal}
Consider an M/G/1 with any nicely light-tailed job size distribution under a SOAP policy.
The policy is log-tail-optimal if $a^* = 0$.
\end{proposition}
\begin{proof}
Suppose that $a^*=0$. In this case, due to the FCFS tiebreaking (\cref{def:soap}), the oldest job in the system always has priority over all other jobs. Therefore the SOAP policy is exactly FCFS, which is known to be log-tail-optimal \citep{stolyar_largest_2001, boxma_tails_2007}.
\end{proof}

\subsection{Tail-Pessimal Case}
\label{sec:light_soap:pessimal}

\Citet{mandjes_sojourn_2005} show that the FB policy, which has rank function $r(a) = a$, is log-tail-pessimal. Their argument focuses on analyzing the response time of very large jobs, showing that such jobs' response time tails have small decay rate. The fact that their argument focuses on just the very large jobs suggests that the essential property of FB is that it assigns the largest rank at the largest ages. Our tail pessimality result below shows this is indeed the case.

\begin{restatable:proposition}\label{lem:light-pessimal}
Consider an M/G/1 with any nicely light-tailed job size distribution under a SOAP policy.
The policy is log-tail-pessimal if $a^* = {\xmax}$.
\end{restatable:proposition}

We prove \cref{lem:light-pessimal} in \cref{sec:computations:light} by generalizing the proof \citet{mandjes_sojourn_2005} give for FB, making use of several of their intermediate results along the way. We describe the approach below, after which we present the proof.

Consider first the case where $\xmax < \infty$, meaning we could have a job of size~$\xmax$. Because $a^* = \xmax$, jobs of size~$\xmax$ complete at the end of the busy period during which they arrive. This is the latest time a job can complete under a work-conserving policy, so $d(T(\xmax)) = d_{\min}$, where $d_{\min}$ is the minimal, and thus pessimal, response time tail decay rate. If $X$ has an atom at~$\xmax$, meaning $X = \xmax$ occurs with positive probability, then $d(T) = d(T(\xmax)) = d_{\min}$.

Of course, for general job size distributions~$X$, it may be that $X = \xmax$ occurs with probability~$0$. This is certainly the case if $\xmax = \infty$. Therefore, to generalize the argument above, instead of considering~$T(\xmax)$, we consider $T(x)$ in the $x \to \xmax$ limit. While $d(T(x)) > d_{\min}$ for any fixed $x < \xmax$, we show that $\lim_{x \to \xmax} d(T(x)) = d_{\min}$. This means that for any $\epsilon > 0$, a positive fraction of jobs experience decay rate less than $d_{\min} + \epsilon$, implying $d(T) = d_{\min}$, as desired.

The last ingredient we need is notation for discussing the M/G/1, and in particular busy periods. This is because $d_{\min}$ turns out to be a busy period's decay rate \citep[Corollary~6]{mandjes_sojourn_2005}.

\begin{definition}
    \label{def:busy_period}
    \leavevmode
    \begin{enumerate:definition}[beginpenalty=10000]
    \item
        We denote by $B$ the distribution of an M/G/1 busy period length with arrival rate~$\lambda$ and job size distribution~$X$. More generally, we write $B(u)$ for a busy period with initial work~$u$.
    \item
        We denote by $B_a$ the distribution of an M/G/1 busy period length with arrival rate~$\lambda$ and truncated job size distribution~$\min\{X, a\}$. More generally, we write $B_a(u)$ for a such busy period with initial work~$u$.
    \item
        We denote by $W$ the distribution of the total amount of work in an M/G/1.
    \end{enumerate:definition}
\end{definition}

It is known that $d_{\min} = d(B)$ \citep[Corollary~6]{mandjes_sojourn_2005}, so proving \cref{lem:light-pessimal} amounts to showing $\lim_{x \to \xmax} d(T(x)) = d(B)$. To show this, we first prove $d(T(x)) \leq d(B_{y_x})$, so it suffices to show $\lim_{x \to \xmax} d(B_{y_x}) = d(B)$. This follows from the known fact that $\lim_{y \to \xmax} d(B_y) = d(B)$ \citep[Proposition~8]{mandjes_sojourn_2005} and additional computation. See \cref{sec:computations:light} for details.

\subsection{Tail-Intermediate Case}
\label{sec:light_soap:intermediate}


We finally turn to the case where $0<a^*<{\xmax}$, where we will show that the corresponding SOAP policy is log-tail-intermediate. We first simplify the problem of analyzing an arbitrary SOAP policy with $0<a^*<\xmax$ to the problem of analyzing two policies, called \emph{Step} and \emph{Spike} (\cref{def:step_spike, fig:step_spike}), with similar rank functions. We then bound the decay rates of Step and Spike.

For brevity, we give only the key definitions and lemma structure of the proof, building up to \cref{lem:light-intermediate}, which states the main result for the tail-intermediate case. The proofs of the individual steps are either largely computational or follow easily from prior work, so we defer the proofs to \cref{sec:computations:light}.

\begin{definition}
    \label{def:step_spike}
    For a given value of $a^* \in (0, \xmax)$ and $r^* > 0$, the \emph{Step} and \emph{Spike} policies are the SOAP policies given by the following rank functions, which are illustrated in \cref{fig:step_spike}:
    \begin{align*}
        r_\step(a) &= r^* \1(a \geq a^*),
        \iftwocoleq{\\}{&}
        r_\spike(a) &= r^* \1(a = a^*).
    \end{align*}
\end{definition}

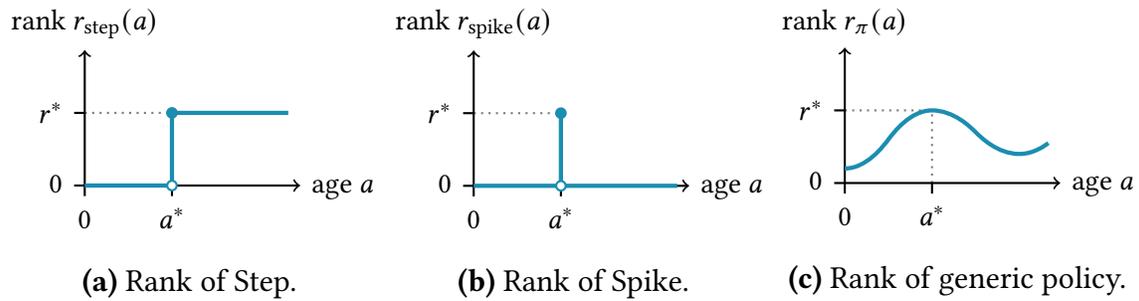
\begin{figure}
    \centering
    \begin{subfigure}[t]{0.33\linewidth}
        \centering
        \input{figures/step}
        \caption{Rank of Step.}
    \end{subfigure}\hfill
    \begin{subfigure}[t]{0.33\linewidth}
        \centering
        \input{figures/spike}
        \caption{Rank of Spike.}
    \end{subfigure}\hfill
    \begin{subfigure}[t]{0.33\linewidth}
        \centering
        \input{figures/generic_intermediate}
        \caption{Rank of generic policy.}
    \end{subfigure}
    \caption{%
        Rank functions of multiple policies for a given worst age~$a^*$ and maximum rank~$r^*$, where $0 < a^* < \xmax$. As shown in \cref{lem:light_soap_step_spike}, for a fixed job size distribution and value of~$a^*$, (a) the Step policy gives the worst possible tail decay, while (b) the Spike policy gives the best possible tail decay. This is because in (a), jobs of age $a^*$ or greater have the worst possible rank~$r^*$, while in (b), jobs of age $a^*$ or greater have the best possible rank~$0$. Under (c) a generic SOAP policy~$\pi$, jobs of age $a^*$ or greater have rank somewhere in between.}
    \label{fig:step_spike}
\end{figure}

We will compare the response time tail of a SOAP policy with $0 < a^* < \xmax$ and $r(a^*) = r^*$ to the response time tails of Step and Spike for those values of $a^*$ and~$r^*$. This comparison is possible because, as illustrated in \cref{fig:step_spike}, all three policies have qualitatively similar rank functions.\footnote{%
    The choice of $r^*$ does not actually affect the scheduling decisions of the Step and Spike policies, but it is convenient in discussion for all three rank functions to have the same maximum.}

For the remainder of this section, we consider SOAP policies~$\pi$ with $0 < a^*< \xmax$ and $r(a^*) = r^*$. We divide jobs in two classes:
\begin{itemize}
    \item \emph{class~1}, jobs of size at most~$a^*$; and
    \item \emph{class~2}, jobs of size greater than~$a^*$.
\end{itemize}
For each class $i\in\{1,2\}$, let $\lambda^{(i)}$ be the arrival rate of class~$i$ jobs, and let $T_\pi^{(i)}$ be the response time of class~$i$ jobs under policy~$\pi$.

It turns out that only class~2 jobs affect the asymptotic decay rate of response time. Moreover, it turns out that the Step and Spike policies represent the worst-case and best-case scenarios for class~2 jobs. These facts, expressed in the following lemma, imply that it suffices to analyze the response time decay rate for class~2 jobs under Step and Spike.

\begin{restatable:lemma}
    \label{lem:light_soap_step_spike}
    Let $\pi$ be a SOAP policy with $0 < a^* < \xmax\esub$. We have
    \begin{equation*}
        d(T_\pi) = d(T^{(2)}_\pi) \in [d(T^{(2)}_\step), d(T^{(2)}_\spike)].
    \end{equation*}
\end{restatable:lemma}

With \cref{lem:light_soap_step_spike} in hand, to show that $\pi$ is tail-intermediate, it suffices to show that both Step and Spike are tail-intermediate by bounding $d(T^{(2)}_\step)$ and $d(T^{(2)}_\spike)$. We begin by characterizing $T^{(2)}_\step$ and $T^{(2)}_\spike$ in terms of the M/G/1 concepts defined in \cref{def:busy_period}.

\begin{restatable:lemma}
    \label{lem:step_spike_class_2}
    The response time distributions of class~2 jobs under Step and Spike are
    \begin{align*}
        T^{(2)}_\step &=_\st B_{a^*}(W) + B_{a^*}(X^{(2)}),
        \iftwocoleq{\\}{&}
        T^{(2)}_\spike &=_\st B_{a^*}(W) + B_{a^*}(a^*) + X^{(2)} - a^*.
    \end{align*}
    where $X^{(2)} = (X \mid X > a^*)$ is the size distribution of class~$2$ jobs, and the random variables in each sum are mutually independent.
\end{restatable:lemma}


Combining \cref{lem:light_soap_step_spike, lem:step_spike_class_2} reduces the question of analyzing the decay rate of~$\pi$ to analyzing the decay rate of the busy periods in \cref{lem:step_spike_class_2}. We will see soon that $B_{a^*}(W)$ is the dominant term, so both Step and Spike have decay rate $d(B_{a^*}(W))$. In order to prove $B_{a^*}(W)$ is indeed the dominant term, and in order to bound its decay rate, we make heavy use of Laplace-Stieltjes transforms (\cref{sec:preliminaries:light}).

Recall from \cref{eq:decay_lst} that one can determine a random variable's decay rate
by determining when its Laplace-Stieltjes transform converges.\footnote{%
    The validity of \cref{eq:decay_lst} for the light-tailed distributions we consider rests on the assumptions we make in \cref{def:light}. See \cref{sec:regularity} for details.}
We therefore introduce notation to describe when a transform converges.
For functions $f : \R \to \R \cup \{-\infty, \infty\}$
which diverge below a certain value and converge above it, define
\begin{align*}
    \gamma(f)
    &= \sup\{s \in \R : |f(s)| = \infty\}
    \iftwocoleq{\\* &}{}
    = \inf\{s \in \R : |f(s)| < \infty\}
\end{align*}
to be the value at which $f$ switches from diverging to converging.
We can thus rewrite \cref{eq:decay_lst} as
\begin{equation}
    \label{eq:decay_gamma}
    d(V) = -\gamma(\lst{V}).
\end{equation}

Because we will be working with Laplace-Stieltjes transforms, we begin by recalling standard results for the transforms of the quantities defined in \cref{def:busy_period}. Let
\begin{align*}
    \sigma^{-1}(s) &= s - \lambda(1 - \lst{X}(s)),
    \iftwocoleq{\\}{&}
    \sigma^{-1}_a(s) &= s - \lambda(1 - \lst{\min\{X, a\}}(s).
\end{align*}
We can define a partial inverse $\sigma$ of $\sigma^{-1}$ such that $\sigma(s)$ is the greatest real solution to
\begin{equation*}
    \sigma(s) = s + \lambda(1 - \lst{X}(\sigma(s))),
\end{equation*}
letting $\sigma(s) = -\infty$ if there is no real solution. We define $\sigma_a$ similarly.

Standard M/G/1 results \citep{harchol-balter_performance_2013} express the Laplace-Stieltjes transforms of random variables in \cref{def:busy_period} using $\sigma$ and~$\sigma_a$. Specifically, for any nonnegative random variable~$U$,
\begin{align*}
    \lst{B(U)}(s) &= \lst{U}(\sigma(s)),
    \iftwocoleq{\\}{&}
    \yestag
    \label{eq:mg1_lst}
    \lst{B_a(U)}(s) &= \lst{U}(\sigma_a(s)),
    \iftwocoleq{\\}{&}
    \lst{W}(s) &= \frac{s(1 - \rho)}{\sigma^{-1}(s)}.
\end{align*}
Therefore, to understand the decay rates of random variables in \cref{def:busy_period}, we need to understand $\sigma^{-1}$, $\sigma$, and~$\sigma_a$. \Fref{fig:sigma} illustrates $\sigma$ and~$\sigma^{-1}$ and some key values associated with them. We make frequent use of relationships between these values and other properties of $\sigma$ and~$\sigma^{-1}$, which are summarized below.

\begin{figure*}
    \centering
    \input{figures/sigma}
    \caption{Illustration of $\sigma^{-1}$ (green curve) and key values associated with it. The partial inverse of $\sigma^{-1}$, namely~$\sigma$, corresponds to to the branch going from the minimum of $\sigma^{-1}$ to the right (orange highlight). As shown in \cref{lem:sigma_si}, we have $s_0 < s_1 < s_2 < s_3$.}
    \label{fig:sigma}
\end{figure*}
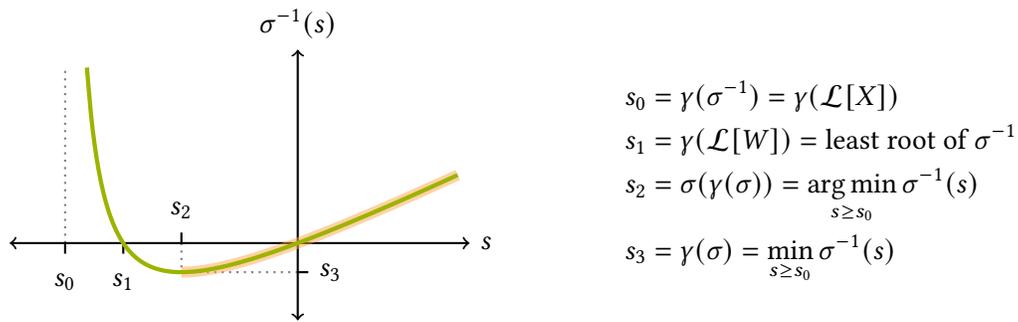

\begin{restatable:lemma}
    \label{lem:sigma_properties}
    Consider an M/G/1 with nicely light-tailed job size distribution~$X$, and define
    \begin{align*}
        s_0 &= \gamma(\sigma^{-1}) = \gamma(\lst{X}),
        \iftwocoleq{
            \\
            s_1 &= \gamma(\lst{W}) = \text{least root of $\sigma^{-1}$}, \\
            s_2 &= \sigma(\gamma(\sigma)) = \argmin_{s \geq s_0} \sigma^{-1}(s), \\
        }{
            &
            s_2 &= \sigma(\gamma(\sigma)) = \argmin_{s \geq s_0} \sigma^{-1}(s), \\
            s_1 &= \gamma(\lst{W}) = \text{least root of $\sigma^{-1}$}, &
        }
        s_3 &= \gamma(\sigma) = \min_{s \geq s_0} \sigma^{-1}(s).
    \end{align*}
    Then, as illustrated in \cref{fig:sigma}, the following hold:
    \begin{enumerate:lemma}
        \item \label{lem:sigma_inv_convex}
        $\sigma^{-1}$ is convex on $(s_0, \infty)$, decreasing on $(s_0, s_2)$, and increasing on $(s_2, \infty)$;
        \item \label{lem:sigma_si}
        $s_0 < s_1 < s_2 < s_3 < 0$.
    \end{enumerate:lemma}
    Analogous statements hold for $\sigma_a\esub$ for all $a \in (0, \xmax]$.
\end{restatable:lemma}

Together, \cref{eq:mg1_lst, lem:sigma_properties} give us the last ingredients we need to compute the decay rate of~$T_\pi$. We then show that this decay rate is neither optimal nor pessimal.

\begin{restatable:lemma}
    \label{lem:light_soap_decay}
    Let $\pi$ be a SOAP policy with $0 < a^* < \xmax\esub$. Then the decay rate of its response time is
    \[
        d(T_\pi) = -\gamma(\lst{W} \circ \sigma_{a^*}).
    \]
\end{restatable:lemma}

\begin{restatable:proposition}\label{lem:light-intermediate}
    Consider an M/G/1 with any nicely light-tailed job size distribution under a SOAP policy. The policy is log-tail-intermediate if $0 < a^* < \xmax\esub$.
\end{restatable:proposition}

\section{Light-Tailed Job Sizes: Gittins can be Tail-Optimal, Tail-Pessimal, or In Between}
\label{sec:light_gittins}


\restate*{thm:light_gittins}

\begin{proof}
    By \cref{thm:light_soap}, it suffices to determine the worst age~$a^*$.
    Combining the following two prior results characterizes~$a^*$
    in terms of whether $X$ is in each of $\NBUE$ and $\ENBUE$.
    \begin{itemize}
        \item A result of \citet[Proposition~7]{aalto_gittins_2009}
        implies $a^* = 0$ if and only if $X \in \NBUE$.
        \item A result of \citet[Proposition~9]{aalto_properties_2011}
        implies $a^* < \xmax$ if and only if $X \in \ENBUE$.
        \qedhere
    \end{itemize}
\end{proof}

\restate*{thm:light_approximate_gittins}

\begin{proof}
    Suppose $\E{X - a \given X > a}$ is uniformly bounded.
    \Cref{def:gittins} implies
    \begin{equation*}
        r_\gittins(a) \leq \frac{\int_a^\infty \tail{t} \d{t}}{\tail{a}} = \E{X - a \given X > a},
    \end{equation*}
    so $r_\gittins(a)$ is also uniformly bounded.
    This means that for any $\epsilon > 0$,
    there exists some sufficiently large age~$a(\epsilon)$ such that
    increasing the rank at age~$a(\epsilon) < \xmax$ from $r_\gittins(a(\epsilon))$
    to $(1 + \epsilon)r_\gittins(a(\epsilon))$
    and leaving all other ranks unchanged
    yields a new SOAP policy with worst age $a^* = a(\epsilon)$.
    By construction, the new policy is a $(1 + \epsilon)$-approximate Gittins policy,
    and because its worst age is $a^* < \xmax$,
    \cref{thm:light_soap} implies it is log-tail-optimal or log-tail-intermediate.
\end{proof}

Recall from \cref{thm:approximate_gittins}
that a $(1 + \epsilon)$-approximate Gittins policy
achieves mean response time within a factor of $1 + \epsilon$ of optimal.
We defer its proof to \cref{sec:gittins}.
This means that \cref{thm:light_approximate_gittins},
whose precondition applies to non-pathological light-tailed job size distributions,
gives a non-tail-pessimal policy with near-optimal mean response time.

\begin{remark}
    \label{rmk:serpt}
    The \emph{Shortest Expected Remaining Processing Time} (SERPT) policy,
    which has rank function $r_\serpt(a) = \E{X - a \given X > a}$,
    is sometimes considered as a simpler alternative to Gittins
    \citep{scully_soap_2018, scully_simple_2020}.
    Our results imply that SERPT has the same M/G/1 tail optimality properties as Gittins for the class of job size distributions we consider.
    \begin{itemize}
    \item
        \Citet{scully_characterizing_2020} show that
        SERPT is always tail-optimal in the heavy-tailed case,
        which matches what we show for Gittins.
    \item
        \Cref{thm:light_soap, def:nbue} imply that
        in the light-tailed case,
        SERPT is tail-optimal, tail-intermediate, and tail-pessimal
        under the same conditions as we show for Gittins
        in \cref{thm:light_gittins}.
        In fact, one can show a stronger property:
        SERPT's and Gittins's response time distributions have the same decay rate.
        This follows from the fact that SERPT and Gittins
        have the same worst age~$a^*$,
        as argued in the proof of \cref{thm:light_gittins}.
    \end{itemize}
\end{remark}

\section{Conclusion}
\label{sec:conclusion}

In this paper we have characterized the asymptotic tail performance of the response time in an M/G/1 queue under very broad conditions, namely for every SOAP policy, and for both heavy- and light-tailed job size distributions.
In the heavy-tailed case, we characterize tail-optimal policies by a sufficient condition on the rank function (\cref{thm:heavy_soap}). This condition holds for a wide range of SOAP policies, and specifically for the Gittins policy (\cref{thm:heavy_gittins}), providing the first proof of its tail optimality under general conditions.
In the light-tailed case, we classify policies' performance as tail-optimal, tail-pessimal, or tail-intermediate. We show that the performance of a SOAP policy depends on the age at which the maximal rank is attained (\cref{thm:light_soap}). It turns out that the Gittins policy may belong to any of the three categories, depending on the job size distribution (\cref{thm:light_gittins}). Finally, when Gittins has pessimal tail performance, boundedness of the expected remaining job size implies that there exists a slight modification of Gittins that has optimal or intermediate tail while maintaining near-optimal mean response time (\cref{thm:light_approximate_gittins}).

\subsection{Returning to the Motivating Questions}
\label{sec:conclusion:questions}

We conclude by returning to \cref{que:simultaneous, que:gittins_tail, que:improve_pessimal}, restated below for convenience.

\restate*{que:simultaneous}

\restate*{que:gittins_tail}

\restate*{que:improve_pessimal}

Our characterization of Gittins's tail asymptotics
(\cref{thm:heavy_gittins, thm:light_gittins})
answers \cref{que:gittins_tail},
and our modification in the case where Gittins is tail-pessimal
(\cref{thm:light_approximate_gittins})
answers \cref{que:improve_pessimal} affirmatively.
This leaves only \cref{que:simultaneous}.
In cases where we have shown that Gittins is tail-optimal,
the answer is clearly affirmative.
We might hope to conclude that the answer is negative
in cases where we have shown that Gittins is tail-pessimal or tail-intermediate,
but the situation is still slightly unclear.
The remaining ambiguity is due to the fact that
we have only considered FCFS tiebreaking
when two jobs have the same rank (\cref{def:soap}),
as we explain in more detail below.

The Gittins policy minimizes mean response time
with \emph{arbitrary tiebreaking} between jobs of the same rank
\citep{gittins_multi-armed_1989, gittins_multi-armed_2011, scully_gittins_2021}.
Moreover, these proofs can be extended to show that any ``non-Gittins'' policy
is \emph{strictly suboptimal} for mean response time,
where a ``non-Gittins'' policy is one that for a non-vanishing fraction of time
serves a job other than one of minimal Gittins rank.
Therefore, to fully answer \cref{que:simultaneous},
one would have to consider Gittins under arbitrary tiebreaking rules.
We conjecture that using a different tiebreaking rule
cannot improve the asymptotic decay rate of Gittins's response time
in the light-tailed case.

Finally, we note that we have of course only answered
\cref{que:simultaneous, que:gittins_tail, que:improve_pessimal}
for the classes of heavy- and light-tailed job size distributions
we consider in this work (\cref{def:heavy, def:light}).
Practically speaking,
we believe the classes of distributions we consider
are likely broad enough to draw useful conclusions.
But it remains an open question whether we may extend our theory
to broader classes of distributions.
In particular, it seems likely that our proofs may hold mostly unchanged
for additional light-tailed job size distributions,
as discussed in \cref{sec:regularity:class_ii}.

\section*{Acknowledgements}

We thank Adam Wierman, Onno Boxma, and Jan-Pieter Dorsman for helpful discussions. We also thank the anonymous referees for helpful comments that significantly improved the presentation.

Ziv Scully conducted this research in part while a graduate student at Carnegie Mellon University, in part while visiting the Simons Institute for the Theory of Computing, and in part while a FODSI postdoc at Harvard and MIT. He was supported by NSF grant nos. CMMI-1938909, CMMI-2307008, CSR-1763701, DMS-2023528, and DMS-2022448.
Lucas van Kreveld was supported by the NWO through Gravitation grant NETWORKS-024.002.003.


\bibliographystyle{ACM-ish}
{\small\setlength{\bibsep}{\smallskipamount}
\bibliography{refs}}






\appendix

\section{Deferred Proofs for Tail Asymptotics of SOAP Policies}
\label{sec:computations}

\subsection{Proofs for Heavy-Tailed Job Sizes}
\label{sec:computations:heavy}

\restate*{lem:X_upper_computation}

\begin{proof}
    We first show~(i). Because $(x, c_0[w_x])$ is a $w_x$-interval,
    \cref{cond:exponents} implies
    \begin{equation}
        \label{eq:X_z_interval}
        c_0[w_x] - x = O(x^{\zeta + \theta}).
    \end{equation}
    We compute
    \begin{align*}
        \iftwocoleqec{\MoveEqLeft}{} \E{X_0[w_x]^{p + 1}} \iftwocoleqec{\\*}{}
        &= \int_0^{c_0[w_x]} (p + 1) t^p \tail{t} \d{t}
        \byref{lem:kth_relevant_moment}
        \\
        &\leq \int_0^{O(x^{\max\{1, \zeta + \theta\}})} O(t^{p - \alpha}) \d{t}
        \byref{def:heavy, eq:X_z_interval}
        \\
        &= \begin{cases}
            O(1) & \text{if } p < \alpha - 1 \\
            O(\log x) & \text{if } p = \alpha - 1 \\
            O(x^{\max\{1, \zeta + \theta\}(p - \alpha + 1)}) & \text{if } p > \alpha - 1,
        \end{cases}
    \end{align*}
    thus proving~(i).

    We now show~(ii), following a similar argument but with a more involved computation. Note that \cref{def:worst_future_rank, def:kth_relevant_age_interval} together imply
    \begin{equation}
        \label{eq:bkwX_geq_x}
        \begin{aligned}
            b_k[w_x] \geq x &&\text{for all } k \geq 1.
        \end{aligned}
    \end{equation}
    We compute
    {\allowdisplaybreaks
    \begin{align*}
        \iftwocoleqec{
            \mathrlap{\sum_{k = 1}^{\kmax[w_x]} \E{X_k[w_x]^{p + 1}}}\quad& \\*
        }{
            \MoveEqLeft \sum_{k = 1}^{\kmax[w_x]} \E{X_k[w_x]^{p + 1}} \\*
        }
        &= \sum_{k = 1}^{\kmax[w_x]} \int_{b_k[w_x]}^{c_k[w_x]} (p + 1) (t - b_k[w_x])^p \tail{t} \d{t}
        \byref{lem:kth_relevant_moment}
        \\
        &\leq \sum_{k = 1}^{\kmax[w_x]} (p + 1) (c_k[w_x] - b_k[w_x])^p \int_{b_k[w_x]}^{c_k[w_x]} \tail{t} \d{t}
        \byref{eq:bkwX_geq_x}
        \\
        &\leq \sum_{k = 1}^{\kmax[w_x]} O(x^{\theta p} \cdot b_k[w_x]^{\zeta p}) \int_{b_k[w_x]}^{c_k[w_x]} O(t^{-\alpha}) \d{t}
        \byref{def:heavy, cond:exponents}
        \\
        &\leq \sum_{k = 1}^{\kmax[w_x]} O(x^{\theta p}) \int_{b_k[w_x]}^{c_k[w_x]} O(t^{\zeta p - \alpha}) \d{t}
        \\
        &\leq O(x^{\theta p}) \int_x^{c_{\kmax[w_x]}[w_x]} O(t^{\zeta p - \alpha}) \d{t}
        \byref{eq:bkwX_geq_x}
        \\
        &\leq O(x^{\theta p}) \int_x^{O(x^\eta)} O(t^{\zeta p - \alpha}) \d{t}
        \byref{cond:exponents}
        \\
        &= \begin{cases}
            O(x^{\theta p + \zeta p - \alpha + 1}) & \text{if } \zeta p < \alpha - 1 \\
            O(x^{\theta p} \log x^\eta) & \text{if } \zeta p = \alpha - 1 \\
            O(x^{\theta p + \eta(\zeta p - \alpha + 1)}) & \text{if } \zeta p > \alpha - 1,
        \end{cases}
    \end{align*}}
    thus proving~(ii).
\end{proof}

\restate*{lem:c0_bound}

\begin{proof}
    Because $\kappa > \theta \geq 0$,
    by \cref{cond:exponents, eq:bkwX_geq_x},
    for all $u \geq 0$ and $k \geq 1$,
    \begin{equation}
        \label{eq:exponents_k}
        u \geq \Omega\gp*{\gp*{\frac{c_k[w_u] - b_k[w_u]}{b_k[w_u]^\zeta}}^{1/\kappa}}.
    \end{equation}
    We now plug in $u = c_0[w_x(a)-]$ and make the following observations.
    \begin{itemize}
        \item
        By \cref{def:kth_relevant_age_interval},
        we know $u = c_0[w_x(a)-]$ is the earliest age at which
        a job has rank at least $w_x(a)$,
        so $w_u = w_x(a)$.
        \item
        By \cref{def:worst_future_rank},
        a job's rank is at most $w_x(a)$ between ages $a$ and~$x$,
        so there exists $k \geq 1$ such that
        \begin{equation*}
            b_k[w_x(a)] \leq a < x \leq c_k[w_x(a)].
        \end{equation*}
        In particular, $x > b_k[w_x(a)]$ and $x - a \leq c_k[w_x(a)] - b_k[w_x(a)]$.
    \end{itemize}
    Applying these observations to \cref{eq:exponents_k}
    with $u = c_0[w_x(a)-]$ yields the desired bound.
\end{proof}

\subsection{Proofs for Light-Tailed Job Sizes}
\label{sec:computations:light}

\restate*{lem:light-pessimal}

\begin{proof}
Since no work-conserving policy has response time decay rate lower than a busy period's decay rate~$d(B)$ \citep[Corollary~6]{mandjes_sojourn_2005}, it suffices to show $d(T) \leq d(B)$.

Recall that $y_x$ denotes the (first) age of the maximum rank in the interval $[0,x]$. Since $T(x)$ is stochastically increasing in~$x$ (\cref{lem:stochastically_increasing}), it holds that $\P{T(x)>t} \geq \P{T(y_x)>t}$ for all $t\geq 0$. Additionally we have that $\P{T(y_x)>t} \geq \P{T_{\text{FB}}(y_x)>t}$ for all $t,x\geq 0$, where $T_{\text{FB}}(x)$ is the response time for a job of size~$x$ under FB. The reason for this last inequality is that a job of size~$y_x$ must wait for all other jobs to receive up to $y_x$ units of service before completing.\footnote{%
    One can give a more formal proof of the inequality using the SOAP analysis \citep{scully_soap_2018}.}
As a result, \citep[Proposition~8]{mandjes_sojourn_2005} implies
\begin{equation*}
    d(T(x)) \leq d(T(y_x)) \leq d(T_{\text{FB}}(y_x)) = d(B_{y_x}).
\end{equation*}
Additionally, up to its last line the proof of \citep[Lemma~9]{mandjes_sojourn_2005} is valid for arbitrary service policies. If $x_0>0$ is such that $\P{X\geq x_0}>0$, we thus find
\begin{align}
    \notag
    d(T)
    &\leq \P{X\geq x_0}^{-1} \int_{x_0}^{\xmax} d(T(x)) \, dF(x) \\
    \label{eq:decay_T_decay_Byx}
    &\leq \P{X\geq x_0}^{-1} \int_{x_0}^{\xmax} d(B_{y_x}) \, dF(x).
\end{align}

Our goal is to show $d(T) \leq d(B)$, or equivalently $d(T) < d(B) + \epsilon$ for all $\epsilon > 0$. By \cref{eq:decay_T_decay_Byx}, it suffices to show that $\lim_{x \to \xmax} d(B_{y_x}) = d(B)$. It is shown in \citep[Lemma 10]{mandjes_sojourn_2005} that $\lim_{x \to \xmax} d(B_x) = d(B)$, so our task is to show that the limit still holds with $y_x$ instead of~$x$.

Consider arbitrary $\epsilon>0$. Because $\lim_{x \to \infty} d(B_x) = d(B)$, there exists $x_0 > 0$ such that $|d(B_x)-d(B)| < \epsilon$ for all $x>x_0$. Because $a^*=\xmax$, there exists $x_1 > x_0$ such that $y_{x_1}=x_1$, and thus $|d(B_{y_{x_1}})-d(B)| < \epsilon$. But $d(B_{y_x})$ is decreasing in~$x$, because $y_x$ is increasing in~$x$, and $B_x$ is stochastically increasing in~$x$. We conclude that for all $x>x_1$, we have $|d(B_{y_x})-d(B)|<\epsilon$. Our choice of $\epsilon > 0$ was arbitrary, so $\lim_{x\to \xmax} d(B_{y_x}) = d(B)$, as desired.
\end{proof}

\restate*{lem:light_soap_step_spike}

\begin{proof}
Clearly, $T_\pi$ is a mixture of $T_\pi^{(1)}$ and~$T_\pi^{(2)}$. \Cref{lem:stochastically_increasing} implies $T_\pi^{(2)} \geq_\st T_\pi^{(1)}$, implying $d(T_\pi) = d(T_\pi^{(2)})$. The same reasoning applies to Step and Spike. The lemma thus follows if we can show
\begin{equation}
    \label{eq:spike_pi_step}
    T^{(2)}_\spike \leq_\st T^{(2)}_\pi \leq_\st T^{(2)}_\step.
\end{equation}

The comparison in \cref{eq:spike_pi_step} follows from a key fact from the SOAP analysis \citep{scully_soap_2018} called the \emph{Pessimism Principle}, which states that the response time of a particular job~J is unaffected if, instead of following the usual rank function, job~J follows its \emph{worst future rank} function (\cref{def:worst_future_rank}). The intuition is that any jobs that will get served ahead of job~J in the future may as well be served ahead of it right now.

\begin{figure*}
    \centering
    \begin{subfigure}[t]{0.33\linewidth}
        \centering
        \input{figures/step_worst}
        \caption{Worst future rank of Step.}
    \end{subfigure}\hfill
    \begin{subfigure}[t]{0.33\linewidth}
        \centering
        \input{figures/spike_worst}
        \caption{Worst future rank of Spike.}
    \end{subfigure}\hfill
    \begin{subfigure}[t]{0.33\linewidth}
        \centering
        \input{figures/generic_intermediate_worst}
        \caption{Worst future rank of generic policy.}
    \end{subfigure}
    \caption{%
        Worst future rank functions (\cref{def:worst_future_rank}, abbreviated w.f.r., dotted magenta curves) of the policies shown in \cref{fig:step_spike}, with the original rank functions (translucent cyan curves) for reference. We show the worst future rank functions for a class~$2$ job of size $x > a^*$.}
    \label{fig:step_spike_worst}
\end{figure*}
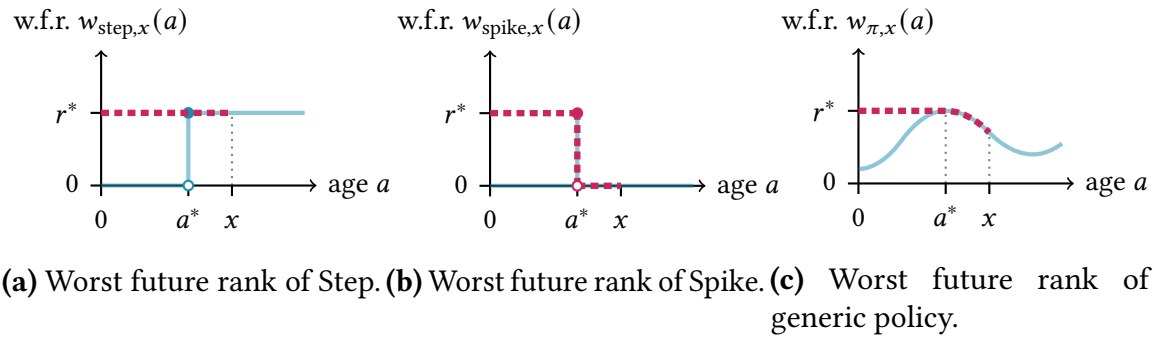

We illustrate in \cref{fig:step_spike_worst} the worst future rank under Step, Spike, and~$\pi$. Notice that, for any size $x > a^*$, we have
\begin{equation*}\begin{aligned}
    \iftwocoleqec{
        \MoveEqLeft \mathrlap{a \in [0, a^*]}\hphantom{a \in (a^*, x)} \ \ \Rightarrow \\*
    }{
        a &\in [0, a^*] & &\Rightarrow &
    }
    r^* &= w_{\spike, x}(a) = w_{\pi, x}(a) = w_{\step, x}(a) = r^*,
    \\
    \iftwocoleqec{
        \MoveEqLeft a \in (a^*, x) \ \ \Rightarrow \\*
    }{
        a &\in (a^*, x) & &\Rightarrow &
    }
    0 &= w_{\spike, x}(a) \leq w_{\pi, x}(a) \leq w_{\step, x}(a) = r^*.
\end{aligned}\end{equation*}
The Pessimism Principle says that we can compute a particular job~J's response time by imagining that it always has its worst future rank. Increasing a job's rank can only increase its response time, so the above worst future rank comparisons imply that for all $x > a^*$,
\begin{equation*}
    T_\spike(x) \leq_\st T_\pi(x) \leq_\st T_\step(x).
\end{equation*}
The desired \cref{eq:spike_pi_step} follows because class~$2$ jobs are those of size greater than~$a^*$.
\end{proof}

\restate*{lem:step_spike_class_2}

\begin{proof}
    This result follows easily from the SOAP analysis \citep{scully_soap_2018}. For completeness, we sketch the main ideas of how the SOAP analysis applies to Step and Spike. Consider a class~2 job~J.
    \begin{itemize}
        \item Under Step, job~J always has worst future rank~$r^*$ (\cref{fig:step_spike_worst}(a)). Job~J is thus delayed by any jobs present when it arrives, plus the pre-age-$a^*$ portion of any jobs that arrive while it is in the system.
        \item Under Spike, job~J has worst future rank~$r^*$ only until age~$a^*$ (\cref{fig:step_spike_worst}(a)). Job~J is thus delayed by any jobs present when it arrives, plus the pre-age-$a^*$ portion of any jobs that arrive before it reaches age~$a^*$. Once job~J reaches age~$a^*$, its worst future rank is~$0$, so no further arrivals delay it.
    \end{itemize}
    The reason in both cases for looking at the pre-age-$a^*$ portion of new arrivals is because at age~$a^*$, those new arrivals reach rank~$r^*$, and thus job~J has priority over them due to FCFS tiebreaking (\cref{def:soap}).

    The delay due to jobs present when job~J arrives corresponds to the~$W$ in each formula, and the delay due to new arrivals corresponds to the $B_{a^*}(\cdot)$ uses. The difference between the formulas is due to the fact that under Step, new arrivals delay job~J until it completes, whereas under Spike, new arrivals delay job~J only if they arrive before it reaches age~$a^*$, with the last $X^{(2)} - a^*$ portion of job~J's service occurring uninterrupted.
\end{proof}

\restate*{lem:sigma_properties}

\begin{proof}
    We prove the statement just for $\sigma$, as the argument for $\sigma_a$ is analogous. The illustration in \cref{fig:sigma} may provide helpful intuition for the arguments that follow.

    We begin by observing some general properties of $\sigma^{-1}$. Because $\lst{X}$ is convex on $(s_0, \infty)$, so is~$\sigma^{-1}$. This, along with the definition of~$s_2$, implies~(a). The slope of $\sigma^{-1}$ at zero is
    \begin{equation*}
        (\sigma^{-1})'(0) = 1 + \lambda \lst{X}'(0) = 1 - \rho \in (0, 1),
    \end{equation*}
    and by \cref{def:light}, we have $\sigma^{-1}(s_0) = \infty > 0$. Additionally, \cref{def:light} implies $s_0 < 0$. This means $\sigma^{-1}$ is negative on a finite nonempty interval, namely $(s_1, 0)$, and nonnegative outside that interval.

    We can now show the inequalities in~(b).
    \begin{itemize}
        \item $s_0 < s_1$: Because $|\sigma^{-1}(s_1)| = 0 < \infty$, we have $s_0 = \gamma(\sigma^{-1}) \leq s_1$. But $\sigma^{-1}(s_0) > 0$, so $s_0 \neq s_1$.
        \item $s_3 < 0$: Because $\sigma^{-1}$ is negative on some interval, its global minimum is negative.
        \item $s_1 < s_2$: Because $s_3 = \sigma^{-1}(s_2) < 0$, we must have $s_2 \in (s_1, 0)$.
        \item $s_2 < s_3$: Because $\sigma^{-1}$ is convex with $\sigma^{-1}(0) = 0$ and $(\sigma^{-1})'(0) \in (0, 1)$, we have $s_2 < \sigma^{-1}(s_2)$.
        \qedhere
    \end{itemize}
\end{proof}

In some of the proofs below, we use the fact that for sums of independent random variables $U, V \geq 0$, \cref{eq:decay_gamma} implies
\begin{align*}
    d(U + V) &= -\gamma(\lst{U + V}) \\
    &= -\max\{\gamma(\lst{U}), \gamma(\lst{V})\} \\
    \yestag
    \label{eq:decay_sum}
    &= \min\{d(U), d(V)\}.
\end{align*}
This is also shown by \citet[Lemma~3]{mandjes_sojourn_2005} without relying on~\cref{eq:decay_gamma}.

\restate*{lem:light_soap_decay}

\begin{proof}
    Combining \cref{lem:light_soap_step_spike, lem:step_spike_class_2}, we have
    \begin{equation*}
        d(T_\pi) \in \sqgp[\big]{d(B_{a^*}(W)), d\gp[\big]{B_{a^*}(W) + B_{a^*}(X^{(2)})}}.
    \end{equation*}
    By \cref{eq:decay_gamma, eq:mg1_lst}, the lower bound is
    \begin{equation*}
        d(B_{a^*}(W)) = -\gamma(\lst{W} \circ \sigma_{a^*}).
    \end{equation*}
    We aim to show that the upper bound matches this. Applying
    \cref{eq:decay_gamma}, \cref{eq:mg1_lst}, and \cref{eq:decay_sum}
    to the upper bound, we see that it suffices to show
    \begin{equation*}
        \gamma(\lst{X^{(2)}} \circ \sigma_{a^*}) \leq \gamma(\lst{W} \circ \sigma_{a^*}).
    \end{equation*}
    \Cref{lem:gamma_sigma_composition}, which we state and prove below, implies the above if $\gamma(\lst{X^{(2)}}) \leq \gamma(\lst{W})$, which in turn is implied by \cref{lem:sigma_si} and the fact that $\gamma(\lst{X^{(2)}}) = \gamma(\lst{X})$.
\end{proof}

The following lemma, which is used in the proof above, relates $\gamma(f \circ \sigma)$ to $\gamma(f)$, thus relating the decay rate of a busy period to the decay rate of its initial work.

\begin{lemma}
    \label{lem:gamma_sigma_composition}
    Let $f : \R \to \R \cup \{-\infty, \infty\}$ be a function for which $\gamma(f)$ is well defined and finite. Then $\gamma(f \circ \sigma)$ is finite, and
    \begin{align*}
        \gamma(f \circ \sigma)
        &= \sigma^{-1}\gp[\big]{\max\curlgp[\big]{\gamma(f), \sigma(\gamma(\sigma))}} \\
        &= \begin{cases}
            \sigma^{-1}(\gamma(f)) & \text{if } \gamma(f) > \sigma(\gamma(\sigma)) \\
            \gamma(\sigma) & \text{otherwise.}
        \end{cases}
    \end{align*}
    In particular, $\gamma(f \circ \sigma)$ is a nondecreasing function of $\gamma(f)$. Analogous statements hold for $\sigma_a\esub$ for all $a \in [0, \xmax]$.
\end{lemma}

\begin{proof}
    We prove the statement just for $\sigma$, as the proof for $\sigma_a$ is analogous. There are two reasons $f(\sigma(s))$ can be infinite:
    \begin{itemize}
        \item We can have $\sigma(s)$ infinite, which happens if and only if $s < \gamma(\sigma)$.
        \item We can have $\sigma(s)$ finite but $f(\sigma(s))$ infinite, which happens if $-\infty < \sigma(s) < \gamma(f)$ and only if $-\infty < \sigma(s) \leq \gamma(f)$.
    \end{itemize}
    Recalling that $\sigma(\gamma(\sigma))$ is the minimum finite value $\sigma(s)$ can take on (see \cref{fig:sigma}), we see that the latter reason can occur for some $s > \gamma(\sigma)$ if and only if $\sigma(\gamma(\sigma)) < \gamma(f)$, implying the desired formula.

    The finiteness of $\gamma(f \circ \sigma)$ follows from finiteness of $\sigma^{-1}(\gamma(f))$ when $\gamma(f) > \gamma(\sigma)$, which by \cref{lem:sigma_si} includes all cases when $\gamma(f) > \sigma(\gamma(\sigma))$. The monotonicity of $\gamma(f \circ \sigma)$ in $\gamma(f)$ follows \cref{lem:sigma_inv_convex}.
\end{proof}

\restate*{lem:light-intermediate}

\begin{proof}
    The optimal decay rate is that of FCFS. A special case of a result of \citet[Theorem~2.2]{stolyar_largest_2001}, together with
    \cref{eq:decay_gamma} and \cref{eq:decay_sum}
    implies this is
    \begin{equation*}
        d(T_{\mathrm{FCFS}}) = d(W + X) = d(W) = -\gamma(\lst{W}).
    \end{equation*}
    The pessimal decay rate is that of FB. A result of \citet[Theorem~1]{mandjes_sojourn_2005} states $d(T_{\mathrm{FB}}) = d(B)$. Together with \cref{eq:decay_gamma, eq:mg1_lst, lem:gamma_sigma_composition}, this implies
    \begin{align*}
        d(T_{\mathrm{FB}}) &= d(B) = -\gamma(\lst{X} \circ \sigma) \\
        &= -\gamma(\sigma) = -\gamma(\lst{W} \circ \sigma).
    \end{align*}
    Above, we use the fact that $\gamma(\lst{X}) < \gamma(\lst{W}) < \sigma(\gamma(\sigma))$, as shown in \cref{lem:sigma_si}, when applying \cref{lem:gamma_sigma_composition}.

    Having computed the optimal and pessimal decay rates in \cref{lem:light_soap_decay}, it suffices to show that in the $0 < a^* < \xmax$ case, we have
    \begin{equation*}
        \gamma(\lst{W}) < \gamma(\lst{W} \circ \sigma_{a^*}) < \gamma(\lst{W} \circ \sigma),
    \end{equation*}
    which we may rewrite as
    \begin{equation*}
        \gamma(\lst{W} \circ \sigma_0) < \gamma(\lst{W} \circ \sigma_{a^*}) < \gamma(\lst{W} \circ \sigma_{\xmax}).
    \end{equation*}
    \Cref{lem:gamma_sigma_a_monotonicity}, which we state and prove below, implies $\gamma(\lst{W} \circ \sigma_a)$ is strictly increasing in~$a$. Therefore, the above holds if $0 < a^* < \xmax$, as desired.
\end{proof}

It remains only to prove the strict monotonicity of $\gamma(\lst{W} \circ \sigma_a)$ in~$a$. We prove a more general statement below.

\begin{lemma}
    \label{lem:gamma_sigma_a_monotonicity}
    Let $f : \R \to \R \cup \{-\infty, \infty\}$ be a function for which $\gamma(f) < 0$, and let $0 \leq a < b \leq \xmax\esub$. Then
    \begin{equation*}
        \gamma(f \circ \sigma_a) < \gamma(f \circ \sigma_b).
    \end{equation*}
\end{lemma}

\begin{proof}
    We begin by comparing $\sigma^{-1}_a(s)$ with $\sigma^{-1}_b(S)$ for all $s < 0$, computing\footnote{%
        Two clarifications about the computation below. First, the notation $U <_\st V$ means that $\P{U > t} \leq \P{V > t}$ for all $t \in \R$, and the set of $t \in \R$ such that $\P{U > t} < \P{V > t}$ has positive Lebesgue measure. Second, because $a < \infty$, the left-hand sides of the last two steps are always finite for all $s < 0$.}
    \begin{equation}
        \begin{alignedat}[b]{2}
            && a &< b \\
            &\Rightarrow\ \ & \min\{X, a\} &<_\st \min\{X, b\} \\
            &\Rightarrow\ \ & \lst{\min\{X, a\}}(s) &< \lst{\min\{X, b\}}(s) \\
            &\Rightarrow\ \ & \sigma^{-1}_a(s) &< \sigma^{-1}_b(s).
        \end{alignedat}
        \label{eq:sigma_inv_a_vs_sigma_inv_b}
    \end{equation}
    There are two important implications of \cref{eq:sigma_inv_a_vs_sigma_inv_b}. The first implication is that the global minimum of $\sigma^{-1}_a$ is less than that of $\sigma^{-1}_b$. But these global minima are $\gamma(\sigma_a)$ and $\gamma(\sigma_b)$, respectively (see \cref{fig:sigma}), so
    \begin{equation}
        \label{eq:gamma_sigma_a_vs_gamma_sigma_b}
        \gamma(\sigma_a) < \gamma(\sigma_b).
    \end{equation}
    This means $\sigma_a(s)$ is finite whenever $\sigma_b(s)$ is. This contributes to the second implication of \cref{eq:sigma_inv_a_vs_sigma_inv_b}: by \cref{lem:sigma_inv_convex},
    \begin{equation}
        \label{eq:sigma_a_vs_sigma_b}
        \sigma_a(s) > \sigma_b(s) \quad \text{for all } s \in [\gamma(\sigma_b), 0).
    \end{equation}
    Note that $f(\sigma_a(s))$ diverges only if $s \leq \gamma(\sigma_a)$ or $\sigma_a(s) \leq \gamma(f)$, while $f(\sigma_b(s))$ diverges if $s < \gamma(\sigma_b)$ or $\sigma_b(s) < \gamma(f)$. Therefore, \cref{eq:gamma_sigma_a_vs_gamma_sigma_b, eq:sigma_a_vs_sigma_b} together imply that there exists a value of~$s$ such that $f(\sigma_b(s))$ diverges while $f(\sigma_a(s))$ does not.
\end{proof}

\section{Properties of the Gittins Policy via the ``Gittins Game''}
\label{sec:gittins}

The goal of this section is to prove two key remaining properties of the Gittins policy,
\cref{thm:approximate_gittins, lem:phi_w-relevant}.
To prove both of these properties,
we will use a different perspective on the Gittins policy called the ``Gittins game''
\citep{scully_gittins_2020}.
The Gittins game gives an alternative way to define the Gittins rank function.
While it is less direct than the definitions we have used so far
(\cref{def:gittins, def:time-per-completion}),
the intermediate steps it introduces turn out to be crucial for proving
\cref{thm:approximate_gittins, lem:phi_w-relevant}.

Aside from \cref{thm:approximate_gittins, lem:phi_w-relevant},
most of the definitions and results in this section are due to \citet{scully_gittins_2020},
who actually study a much more general job model than ours.
For simplicity, we restate the key definitions and results in our setting.
However, the statements and proofs of \cref{thm:approximate_gittins, lem:phi_w-relevant}
are straightforward to translate to
the more general job model of \citet{scully_gittins_2020}.

\subsection{The Gittins Game}

The Gittins game is an optimization problem.
Its inputs are a job at some age~$b$ and a \emph{penalty}~$w$.
During the game, we serve the job for as long as we like.
If the job completes, the game ends.
At any moment before the job completes, we may choose to give up,
in which case we pay the penalty~$w$ and the game immediately ends.
The goal of the game is to minimize the expected sum of
the time spent serving the job plus the penalty paid.

We can think of the Gittins game with penalty~$w$
as an optimal stopping problem whose state is the age~$b$ of the job.
Standard optimal stopping theory \citep{shiryaev_optimal_2008, peskir_optimal_2006}
implies that the optimal strategy thus has the following form:
serve the job until it reaches some age $c \geq b$, then give up.
A possible policy here is never giving up, which is represented by $c = \infty$.

Suppose we start serving a job at age~$b$ and stop if it reaches age~$c$.
The expected amount of time we spend serving the job is
\begin{equation*}
    \service(b, c) = \E{\min\{S, c\} \given S > b} = \int_b^c \frac{\tail{t}}{\tail{b}} \d{t},
\end{equation*}
and the probability the job finishes before reaching age~$b$ is
\begin{equation*}
    \done(b, c) = \P{S \leq c \given S > b} = 1 - \frac{\tail{c}}{\tail{b}}.
\end{equation*}
We can write the time-per-completion function
as $\phi(b, c) = \service(b, c)/\done(b, c)$ (see \cref{def:time-per-completion}).

Suppose we employ the stop-at-age-$c$ policy in the Gittins game
starting from age~$b$ with penalty~$w$.
The expected cost of the Gittins game with this policy is
\begin{equation*}
    \game(w; b, c) = \service(b, c) + w (1 - \done(b, c))
\end{equation*}
The \emph{optimal} cost of the Gittins game is therefore
\begin{equation*}
    \game^*(w; b) = \inf_{c \geq b} \game(w; b, c).
\end{equation*}
The lemma below follows immediately from the definition of $\game^*(w; b)$
as an infimum of $\game(w; b, c)$, each of which is a linear function of~$w$
\citep[Lemmas~5.2 and~5.3]{scully_gittins_2020}.

\begin{lemma}
    \label{lem:game_bound}
    For all ages~$b$,
    the optimal cost $\game^*(w; b)$ is increasing and concave as a function of~$w$.
    Because giving up immediately is always a possible policy,
    it is also bounded above by $\game^*(w; b) \leq w$.
\end{lemma}

\subsection{Relating the Gittins Game to the Gittins Rank Function}

The Gittins game is intimately connected to the Gittins rank function,
and it is this connection that is important for proving \cref{lem:phi_w-relevant}.
The following lemmas state two such connections.
They are the same or very similar
to many previous results in the literature on Gittins in the M/G/1
\citep{scully_gittins_2020, scully_optimal_2018, gittins_multi-armed_1989, gittins_multi-armed_2011, aalto_gittins_2009},
but we sketch their proofs for completeness.

\begin{lemma}
    \label{lem:game_rank}
    The Gittins rank function can be expressed in terms of the Gittins game as
    \begin{align*}
        r(a)
        &= \inf\curlgp{w \geq 0 \given \game^*(w; a) < w} \\
        &= \max\curlgp{w \geq 0 \given \game^*(w; a) = w}.
    \end{align*}
\end{lemma}

\begin{proof}
    The infimum and maximum are equivalent by \cref{lem:game_bound}.
    The infimum is equal to the rank $r(a) = \inf_{c > a} \phi(a, c)$ because,
    by the fact that we can write $\game(w; b, c)$ as
    \begin{equation}
        \label{eq:game_phi}
        \game(w; b, c) = w - (w - \phi(b, c)) \done(b, c),
    \end{equation}
    we have $\game(w; b, c) < w$ if and only if $\phi(b, c) < w$.\footnote{%
        Recall that $\done(b, c) \in [0, 1]$
        and that if $\done(b, c) = 0$,
        then $\phi(b, c) = \infty$ (\cref{def:time-per-completion}).}
\end{proof}

\begin{lemma}
    \label{lem:game_give_up}
    In the Gittins game with penalty~$w$ with the job currently at age~$a$,
    it is optimal to continue serving the job if and only if $r(a) \leq w$,\footnote{%
        Strictly speaking, it is optimal to continue serving the job
        if and only if the rank is upper bounded in a ``forward neighborhood'' of~$a$,
        meaning there exists $\epsilon > 0$
        such that for all $\delta \in [0, \epsilon)$, we have $r(a + \delta) \leq w$.
        For non-pathological job size distributions,
        this holds in the $r(a) < w$ case \citep{aalto_properties_2011},
        so it only needs to be checked when $r(a) = w$.}
    and it is optimal to give up if and only if $r(a) \geq w$.
\end{lemma}

\begin{proof}
    Giving up incurs cost~$w$,
    so by the maximum in \cref{lem:game_rank}, it is optimal to give up if and only if $r(a) \geq w$.
    This means it is optimal to continue serving the job if $r(a) < w$.
    The fact that serving is optimal in the $r(a) = w$ edge case follows from the fact that
    if $\phi(a, c) = w$ for some $c > a$,\footnote{%
        The $c > a$ restriction is why we need the rank to be bounded
        not just at~$a$ but in a forward neighborhood of~$a$.}
    then by \cref{eq:game_phi}, we have $\game(w; a, c) = w$.
\end{proof}

We are now ready to prove \cref{lem:game_give_up}, which we restate below.
Recall that a $w$-interval is one in which the Gittins rank is bounded above by~$w$.
The key to the proof is that
\cref{lem:game_give_up} relates $w$-intervals to optimal play the Gittins game.

\restate*{lem:phi_w-relevant}

\begin{proof}
    \label{pf:phi_w-relevant}
    Consider playing the Gittins game starting from age~$b$.
    By \cref{lem:game_give_up}, giving up if the job reaches age~$c$ is an optimal policy.
    Specifically, because $(b, c)$ is a $w$-interval,
    it is optimal to continue serving the job until at least age~$c$,
    and because the $w$-interval is right-maximal,
    it is optimal to give up if the job reaches age~$c$
    (which never happens if $c = \xmax$).
    This means $\game^*(w; b) = \game(w; b, c)$.
    Combining \cref{lem:game_bound, eq:game_phi} implies $\phi(b, c) \leq w$.
\end{proof}

We note that \cref{lem:phi_w-relevant} is similar, but not identical,
to properties of Gittins in the M/G/1
studied by \citet{aalto_gittins_2009, aalto_properties_2011}.
Related properties have also been shown
for versions of Gittins in discrete-time settings
\citep{gittins_multi-armed_1989, gittins_multi-armed_2011, dumitriu_playing_2003}.

\subsection{Relating the Gittins Game to Mean Response Time}

It remains only to prove \cref{thm:approximate_gittins},
which bounds the mean response time of $q$\=/approximate Gittins policies.
To do so, we use a result of \citet{scully_gittins_2020}
that relates the Gittins game to a system's mean response time.

\begin{definition}
    \label{def:relevant_work}
    Let $r : [0, \xmax) \to \R$ be the rank function of some SOAP policy, and let $w \in \R$.
    \begin{enumerate:definition}
        \item
        The \emph{$(r, w)$-relevant work of a job} is the amount of service the job requires
        to either complete or reach rank at least~$w$ according to~$r$,
        meaning reaching an age~$a$ satisfying $r(a) \geq w$.
        \item
        The \emph{$(r, w)$-relevant work of the system} is the total $(r, w)$-relevant work of all jobs present.
        We denote the steady-state distribution of $(r, w)$-relevant work under policy~$\pi$ by~$W_{\pi}(r, w)$.
        Note that $r$ need not be the rank function of policy~$\pi$.
    \end{enumerate:definition}
\end{definition}

The $(r_\gittins, w)$-relevant work of a job is related to the Gittins game via \cref{lem:game_give_up}:
it is the amount of time we would serve the job
when optimally playing the Gittins game with penalty~$w$.
It turns out that mean $(r_\gittins, w)$-relevant work directly translates into mean response time.

\begin{lemma}[\textnormal{\citet[Theorem~6.3]{scully_gittins_2020}}]
    \label{lem:relevant_work_response_time}
    Under any nonclairvoyant scheduling policy~$\pi$,
    the mean response time can be written in terms of $(r_\gittins, w)$-relevant work as
    \begin{equation*}
        \E{T_\pi} = \frac{1}{\lambda} \int_0^\infty \frac{\E{W_\pi(r_\gittins, w)}}{w^2} \d{w}.
    \end{equation*}
\end{lemma}

With \cref{lem:relevant_work_response_time} in hand,
the proof of \cref{thm:approximate_gittins}, restated below,
reduces to bounding the mean amount of $(r_\gittins, w)$-relevant work
under $q$\=/approximate Gittins policies.

\restate*{thm:approximate_gittins}

\begin{proof}
    \label{pf:approximate_gittins}
    Recall from \cref{def:approximate_gittins} that
    we may assume $r_\pi(a) / r_\gittins(a) \in [1, q]$ for all ages~$a$
    without loss of generality.
    We will prove
    \begin{align*}
        \E{W_\pi(r_\gittins, w)}
        \iftwocoleqec{&}{} \leq \E{W_\pi(r_\pi, qw)} \iftwocoleqec{\\}{}
        \yestag
        \label{eq:relevant_work_chain}
        &\leq \E{W_\gittins(r_\gittins, w)},
    \end{align*}
    from which the theorem follows by the computation below:
    \begin{align*}
        \E{T_\pi}
        &= \frac{1}{\lambda} \int_0^\infty \frac{\E{W_\pi(r_\gittins, w)}}{w^2} \d{w}
        \byref{lem:relevant_work_response_time} \\
        &\leq \frac{1}{\lambda} \int_0^\infty \frac{\E{W_\gittins(r_\gittins, qw)}}{w^2} \d{w}
        \byref{eq:relevant_work_chain} \\
        &= \frac{1}{\lambda} \int_0^\infty \frac{\E{W_\gittins(r_\gittins, w')}}{(w'/q)^2} \d{(w'/q)}
        \eqsidetext{by substituting $w' = qw$} \\
        &= q \E{T_\gittins}.
        \byref{lem:relevant_work_response_time}
    \end{align*}

    To show the left-hand inequality of \cref{eq:relevant_work_chain},
    it suffices to show that an arbitrary job's $(r_\gittins, w)$-relevant work
    is upper bounded by its $(r_\pi, qw)$-relevant work (\cref{def:relevant_work}).
    This is indeed the case:
    $r_\gittins(a) \leq w$ implies $r_\pi(a) \leq q r_\gittins(a) \leq qw$,
    so the job will reach rank~$w$ under Gittins
    after at most as much service as it needs to reach rank~$qw$ under~$\pi$.

    To show the right-hand inequality of \cref{eq:relevant_work_chain}
    we use a property of SOAP policies due to
    \citet[proof of Lemma~5.2]{scully_soap_2018}.
    The property implies that for any rank~$w$ and SOAP policy~$\pi$,
    we can express $\E{W_\pi(r_\pi, w)}$ in terms of just
    the job size distribution~$X$, arrival rate~$\lambda$,
    and the set of ages $A_\pi[w] = \{a \in [0, x_{\max}) \given r_\pi(a) < w\}$.\footnote{%
        \Citet{scully_soap_2018} actually focus on
        $\E{W_\pi(r_\pi, w+)} = \lim_{w' \downarrow w} \E{W_\pi(r_\pi, w')}$
        as opposed to $\E{W_\pi(r_\pi, w)}$,
        but the same reasoning applies to $\E{W_\pi(r_\pi, w+)}$.}
    In particular, for any fixed job size distribution, arrival rate, and rank~$w$,
    $\E{W_\pi(r_\pi, w)}$ is a nondecreasing function of $A_\pi[w]$,
    where we order sets by the usual subset partial ordering.
    We have $r_\pi(a) \geq r_\gittins(a)$,
    which means $A_\pi[w] \subseteq A_\gittins[w]$,
    which implies the right-hand inequality of~\cref{eq:relevant_work_chain}, as desired.
\end{proof}

We note that one can use the techniques of \citet{scully_soap_2018-1}
to generalize the statement and proof of \cref{thm:approximate_gittins}
beyond SOAP policies.
It turns out that \cref{thm:approximate_gittins} still holds
even if we allow $q$\=/approximate Gittins policies
to \emph{adversarially} assign ranks to jobs,
provided that the assigned ranks are still within a factor-$q$ window around
the rank Gittins would assign.

\section{Relationship Between Decay Rate and Laplace-Stieltjes Transform}
\label{sec:regularity}

The goal of this appendix is to justify our computation of decay rates
(\cref{def:decay_rate})
by means of Laplace-Stieltjes transform convergence
(\cref{sec:light_soap:intermediate}).
Our specific goal is to justify our use of~\cref{eq:decay_gamma},
which states $d(V) = -\gamma(\lst{V})$.
As a reminder,
\begin{align*}
    d(V) &= \lim_{t \to \infty} \frac{-\log \P{V > t}}{t}, \\
    \gamma(f) &= \inf\curlgp{s \in \R \given |f(s)| < \infty}.
\end{align*}

\subsection{Sufficient Condition for Computing Decay Rates}

Our main tool for translating between $d(V)$ and $\gamma(\lst{V})$
is a result of \citet{mimica_exponential_2016},
restated as \cref{lem:decay_gamma} below,
which gives a sufficient condition for $d(V) = -\gamma(\lst{V})$.
The result rests on the following definition.

\begin{definition}
    \label{def:regular_variation_at}
    We say a function $f : \R \to \R \cup \{-\infty, \infty\}$
    is \emph{regularly varying from the right at~$s^*$ with negative index},
    or simply ``regularly varying at~$s^*$'',
    if there exists $\alpha > 0$ such that for all $c > 0$,
    \begin{equation*}
        \lim_{s \downarrow 0} \frac{f(s^* + cs)}{f(s^* + s)} = c^{-\alpha}.
    \end{equation*}
    In particular, $f$ having a pole of finite order at~$s^*$ suffices.
\end{definition}

It turns out being regularly varying at the singularity is the condition we need
to express decay rate in terms of Laplace-Stieltjes transform convergence.

\begin{lemma}[\textnormal{special case of \citet[Corollary~1.3]{mimica_exponential_2016}}]
    \label{lem:decay_gamma}
    Let $V$ be a non-negative random variable with $\gamma(\lst{V}) > -\infty$.
    If either $\lst{V}$ or $\lst{V}'$ is regularly varying at~$\gamma(\lst{V})$, then
    \begin{equation*}
        d(V) = -\gamma(\lst{V}).
    \end{equation*}
\end{lemma}

\subsection{Showing the Sufficient Condition for Computing Decay Rates Holds}

It remains to show that the precondition of \cref{lem:decay_gamma} holds
whenever we apply \cref{eq:decay_gamma} in \cref{sec:light_soap:intermediate}.
It turns out that all of the Laplace-Stieltjes transforms to which we apply \cref{eq:decay_gamma}
have a common form,
so we will show that \cref{lem:decay_gamma} applies to all functions of that form.
To describe the form, we need the following definition.

\begin{definition}
    \label{def:sigma_etc}
    Consider an M/G/1 with arrival rate~$\lambda$,
    job size distribution~$X$,
    and load $\rho = \lambda \E{X}$.
    \begin{enumerate:definition}
    \item
        We define the function
        \begin{equation*}
            \sigma_X^{-1}(s) = s - \lambda(1 - \lst{X}(s)).
        \end{equation*}
        Note that $\sigma_X^{-1}(s) = \infty$ if and only if $\lst{X}(s) = \infty$.
    \item
        We define $\sigma_X$ to be the the inverse of $\sigma_X^{-1}$,
        choosing the branch that passes through the origin.
        That is, for $s \geq \inf_r \sigma_X^{-1}(r)$,
        we define $\sigma_X(s)$ to be the greatest real solution to
        \begin{equation*}
            \sigma_X(s) = s + \lambda(1 - \lst{X}(\sigma_X(s))).
        \end{equation*}
        If $s < \inf_r \sigma_X^{-1}(r)$,
        then no such solution exists,
        so we define $\sigma_X(s) = -\infty$.
    \item
        We define the \emph{work-in-system transform}
        \begin{equation*}
            \lst{W_X}(s) = \frac{s(1 - \rho)}{\sigma_X^{-1}(s)}.
        \end{equation*}
    \end{enumerate:definition}
    Note that all of the above definitions depend on both $\lambda$ and~$X$,
    However, because the following discussion considers a fixed arrival rate $\lambda$
    and varies only the job size distribution~$X$,
    we keep $\lambda$ implicit to reduce clutter.
    Additionally, we assume in all uses of the above definitions that $\rho < 1$.
\end{definition}

One may recognize the functions defined in \cref{def:sigma_etc}
as core to the theory of the M/G/1 with job size distribution~$X$
\citep{harchol-balter_performance_2013}.
\begin{itemize}
    \item The work-in-system transform is, as suggested by its name,
    the Laplace-Stieltjes transform of the equilibrium distribution~$W_X$ of
    the total workload in the M/G/1.
    \item The function~$\sigma_X$ is related to busy periods in the M/G/1.
    Specifically, the length of a busy period started by initial workload~$V$
    has Laplace-Stieltjes transform $\lst{V}(\sigma_X(s))$.
\end{itemize}

It turns out that throughout \cref{sec:light_soap:intermediate},
all of the Laplace-Stieltjes transforms
to which we apply \cref{eq:decay_gamma} are of the form $\lst{W_X}$
or $\lst{W_X} \circ \sigma_Y$,
the latter meaning $s \mapsto \lst{W_X}(\sigma_Y(s))$,
for nicely light-tailed job size distributions $X$ and~$Y$ (\cref{def:light}).
Specifically, $X$ is the system's job size distribution,
and $Y$ is either $X$ or a truncation $\min\{X, a^*\}$.
Therefore, to justify the uses of \cref{eq:decay_gamma} using \cref{lem:decay_gamma},
it suffices to prove \cref{lem:work_rv, lem:work_sigma_rv} below.\footnote{%
    While \cref{def:sigma_etc} assumes a single arrival rate~$\lambda$,
    \cref{lem:work_sigma_rv} easily generalizes to the case where
    $\lst{W_X}$ and $\sigma_Y$ are defined using different arrival rates.}

\begin{proposition}
    \label{lem:work_rv}
    For any nicely light-tailed job size distribution~$X$,
    \begin{enumerate:proposition}
        \item $\gamma(\lst{W_X}) \in (-\infty, 0)$; and
        \item $\lst{W_X}\esub$ has a first-order pole at $\gamma(\lst{W_X})$,
        so it is regularly varying at $\gamma(\lst{W_X})$.
    \end{enumerate:proposition}
\end{proposition}

\begin{proposition}
    \label{lem:work_sigma_rv}
    For any nicely light-tailed job size distributions $X$ and~$Y$,
    \begin{enumerate:proposition}
        \item $\gamma(\lst{W_X} \circ \sigma_Y) \in (-\infty, 0)$, and
        \item either $\lst{W_X} \circ \sigma_Y\esub$ or $(\lst{W_X} \circ \sigma_Y)'$ is regularly varying at $\gamma(\lst{W_X} \circ \sigma_Y)$.
    \end{enumerate:proposition}
\end{proposition}

Our approach is as follows.
We first prove \cref{lem:work_rv}.
We then prove a lemma characterizing~$\sigma_X$,
which we use in conjunction with \cref{lem:work_rv} to prove \cref{lem:work_sigma_rv}

\begin{proof}[Proof of \cref{lem:work_rv}]
    Recall from \cref{def:sigma_etc} that $\lst{W_X}(s) = s(1 - \rho)/\sigma_X^{-1}(s)$,
    so we focus on $\sigma_X^{-1}$.
    Because $\lst{X}$ is a mixture of exponentials, $\sigma_X^{-1}$ is convex,
    so it has at most two real roots.
    It is well-known that under the assumption on~$X$ made in \cref{def:light},
    $\sigma_X^{-1}$ has a first-order root at~$0$
    and a negative first-order root \citep{abate_asymptotics_1997, mandjes_sojourn_2005},
    the latter of which is $\gamma(\lst{W_X})$,
    but we give a brief proof for completeness.
    One can compute $\sigma_X^{-1}(0) = 0$
    and $(\sigma_X^{-1})'(0) = 1 - \rho$, so $\sigma_X^{-1}$ has a first-order root at~$0$.
    \Cref{def:light} implies $\lst{X}(\gamma(\lst{X})) = \infty$,
    so $\sigma_X^{-1}(\gamma(\lst{X})) = \infty$.
    This means $\sigma_X^{-1}$ has another first-order root in $(\gamma(\lst{X}), 0)$.
\end{proof}

\begin{lemma}
    \label{lem:sigma_rv}
    For any nicely light-tailed job size distribution~$X$,
    \begin{enumerate:lemma}
        \item $\gamma(\sigma_X) \in (-\infty, 0)$;
        \item $\sigma_X(\gamma(\sigma_X)) \in (-\infty, 0)$; and
        \item there exist $C_0, C_1 > 0$ such that in the $s \downarrow 0$ limit,
        \begin{align*}
            \sigma_X(\gamma(\sigma_X) + s) &= \sigma_X(\gamma(\sigma_X)) + C_0\sqrt{s} \pm \Theta(s), \\
            \sigma_X'(\gamma(\sigma_X) + s) &= \frac{C_1}{\sqrt{s}} \pm \Theta(1),
        \end{align*}
        so $\sigma_X'\esub$ is regularly varying at $\gamma(\sigma_X)$.
    \end{enumerate:lemma}
\end{lemma}

\begin{proof}
    As in the proof of \cref{lem:work_rv},
    we again use the fact that $\sigma_X^{-1}$ is convex,
    has roots at a negative number and at zero,
    and is negative between its roots.
    Specifically, this fact implies that $\sigma_X^{-1}$
    has a finite negative global minimum.
    By \cref{def:sigma_etc}, this minimum is $\gamma(\sigma_X)$,
    and the value at which the minimum is attained is $\sigma_X(\gamma(\sigma_X))$
    proving (i) and~(ii).

    It remains only to prove~(iii).
    The fact that Laplace-Stieltjes transforms
    are analytic in the interior of their domains of convergence
    implies that $\sigma_X^{-1}$ can be written as a Taylor series about $\gamma(\sigma_X)$
    whose first nonzero coefficient is quadratic,
    i.e. for some constant $K > 0$,
    \begin{equation*}
        \sigma_X^{-1}(s) = Ks^2 \pm \Theta(s^3).
    \end{equation*}
    An extension of the Lagrange inversion theorem \citep[\S 1.10(vii)]{olver_nist_2021}
    implies that the inverse of $\sigma_X^{-1}$, namely $\sigma_X$,
    may thus be written in the desired form.
    The desired form for $\sigma_X'$, which completes~(iii), then follows from
    \begin{align*}
        (\sigma_X^{-1})'(s) &= 2Ks \pm \Theta(s^2), \\
        \sigma_X'(s) &= \frac{1}{(\sigma^{-1})'(\sigma_X(s))}.
        \qedhere
    \end{align*}
\end{proof}

\begin{proof}[Proof of \cref{lem:work_sigma_rv}]
    There are three cases to consider:
    \begin{itemize}
        \item $\gamma(\lst{W_X}) > \sigma_Y(\gamma(\sigma_Y))$,
        \item $\gamma(\lst{W_X}) < \sigma_Y(\gamma(\sigma_Y))$, and
        \item $\gamma(\lst{W_X}) = \sigma_Y(\gamma(\sigma_Y))$.
    \end{itemize}
    For an intuitive grasp of these cases,
    it is helpful to imagine decreasing $s$ starting at $s = 0$,
    tracking the behavior of $\lst{W_X}(\sigma_Y(s))$ as $s$ decreases.

    If $\gamma(\lst{W_X}) > \sigma_Y(\gamma(\sigma_Y))$,
    then at some point before $s = s^*$ reaches~$\gamma(\sigma_Y)$,
    meaning for some $s^* \in (-\gamma(\sigma_Y), 0)$,
    we have $\gamma(\lst{W_X}) = \sigma_Y(s^*)$.
    This means $\gamma(\lst{W_X} \circ \sigma_Y) = s^*$.
    The Lagrange inversion theorem \citep[\S 1.10(vii)]{olver_nist_2021}
    and the fact that $s > \gamma(\sigma_Y)$
    imply that $\sigma_Y$ can be linearly approximated near~$s^*$,
    so the result follows from \cref{lem:work_rv}.

    If $\gamma(\lst{W_X}) < \sigma_Y(\gamma(\sigma_Y))$,
    then in contrast to the previous case,
    $s$ reaches $\gamma(\sigma_Y)$, the last point at which $\sigma_Y(s)$ is finite,
    before $\sigma(s)$ reaches the pole of $\lst{W_x}$.
    This means $\gamma(\lst{W_X} \circ \sigma_Y) = \gamma(\sigma_Y)$
    Similarly to the previous case, we can linearly approximate $\lst{W_x}$
    near $\gamma(\sigma_Y)$,
    so the result follows from \cref{lem:sigma_rv}.

    If $\gamma(\lst{W_X}) = \sigma_Y(\gamma(\sigma_Y))$,
    then roughly speaking, both of the previous cases' events happen simultaneously:
    just as $s$ reaches $\gamma(\sigma_Y)$, the last point at which $\sigma_Y(s)$ is finite,
    $\sigma_Y(s)$ reaches the pole of $\lst{W_X}$.
    Combining \cref{lem:work_rv, lem:sigma_rv} implies that in the $s \downarrow \gamma(\sigma_Y)$ limit,
    we can approximate $\lst{W_X}(\sigma_Y(s))$ as
    \begin{align*}
        \lst{W_X}(\sigma_Y(\gamma(s))
        \iftwocoleqec{&}{} = \frac{K_0}{\sigma_Y(\gamma(s))} \pm \Theta(1) \iftwocoleqec{\\}{}
        &= \frac{K_1}{\sqrt{s - \gamma(\sigma_Y)}} \pm \Theta(1)
    \end{align*}
    for some constants $K_0, K_1 > 0$,
    from which the result follows.
\end{proof}

\subsection{Expanding the Definition of Nicely Light-Tailed Job Size Distributions}
\label{sec:regularity:class_ii}

The class of light-tailed distributions we consider in \cref{def:light},
namely what \citet{abate_asymptotics_1997} call ``Class~I'' distributions,
is well behaved enough for \cref{lem:work_rv, lem:work_sigma_rv} to hold.
More generally, our results apply to any job size distribution
with positive decay rate
for which one can show \cref{lem:work_rv, lem:work_sigma_rv}.
In particular,
this includes many distributions that \citet{abate_asymptotics_1997} call ``Class~II''.
These are job size distributions~$X$ such that $\lst{X}(\gamma(\lst{X})) < \infty$.

In order to prove \cref{lem:work_rv, lem:work_sigma_rv} for Class~II job size distributions,
one would need to assume a regularity condition.
We believe it would suffice to assume
that $\lst{X}'$ is regularly varying at $\gamma(\lst{X})$.
The main change to the proofs would be additional casework.
For example, it may be that $\lst{W_X}$ still has a first-order pole,
or it may be that it diverges without a pole because $\lst{X}$ does.
See \citet{abate_asymptotics_1997} and references therein for additional discussion.

More generally, it likely suffices to assume that some higher-order derivative $\lst{X}^{(n)}$
is regularly varying at $\gamma(\lst{X})$,
as the result of \citet[Corollary~1.3]{mimica_exponential_2016} we use
applies to higher derivatives as well.
Other results of \citet{mimica_exponential_2016}
may allow one to relax the assumption even further.

\end{document}

%% file: preamble.tex
\usepackage{endnotes}

\useamsalign
\crefname{footnote}{footnote}{footnotes}

\input{figures/macros}

\newcommand{\IfTwoCol}{\ifdefstring{\acmformat}{sigconf}}
\newcommand{\TwoColMoveEqLeft}[1]{%
    \IfTwoCol{\MoveEqLeft}{} #1 \IfTwoCol{\\}{}}
\newcommand{\TwoColBreak}[1][]{\IfTwoCol{\\ &\qquad #1}{}}

\newcommand{\xmax}{x_{\max}}
\newcommand{\kmax}{K}

\newcommand{\ostrict}{\check{o}}

\newcommand{\lst}[1]{\mathcal{L}[#1]}

\newcommand{\st}{\mathrm{st}}
\newcommand{\gittins}{\mathrm{Gtn}}
\newcommand{\serpt}{\mathrm{SERPT}}

\newcommand{\step}{\mathrm{step}}
\newcommand{\spike}{\mathrm{spike}}

\newcommand{\textclass}{\mathsf}
\newcommand{\NBUE}{\textclass{NBUE}}
\newcommand{\ENBUE}{\textclass{ENBUE}}

\newcommand{\service}{\mathop{}\!\mathsf{service}}
\newcommand{\done}{\mathop{}\!\mathsf{done}}
\newcommand{\game}{\mathop{}\!\mathsf{game}}

\newcommand{\tail}[2][]{\appopt{\subopt{\ol{F}}{#1}}{#2}}

\newcommand{\proofin}[2][]{%
    \ifempty{#1}{%
        \begin{proof}%
    }{%
        \IfSubStr{#1}{,}{%
            \begin{proof}[Proofs of \cref{#1}]%
        }{%
            \begin{proof}[Proof of \cref{#1}]%
        }%
    }
    See \cref{#2}. \noqed
  \end{proof}}

\newcommand{\eqsidetextsepsameline}{&&}
\newcommand{\eqsidetext}[2][\eqsidetextsepdefault]{%
  #1 \text{\footnotesize[#2]}}
\newcommand{\byref}[2][\eqsidetextsepdefault]{\eqsidetext[#1]{by \cref{#2}}}

\newtoggle{twocoleq}\togglefalse{twocoleq}
\newcommand{\iftwocoleq}{\iftoggle{twocoleq}}
\newtoggle{twocoleqec}\togglefalse{twocoleqec}
\newcommand{\iftwocoleqec}{\iftoggle{twocoleqec}}

\makeatletter
\patchcmd{\@bspreauthor}{\large}{}{}{}
\patchcmd{\@bspredate}{\large}{}{}{}
\patchcmd{\abstract}{\bfseries}{\@titlestyle}{}{}
\makeatother

%% file: figures/w_x-intervals.tex
\begin{tikzpicture}[figure]
  \fill[orange!35!white] (0, 0) rectangle (6.75, 6);
  \draw[orange!70!white, thin] (0, 6) -- (6.75, 6) -- ++(0, -6);

  \fill[orange!35!white] (8.5, 0) rectangle (10.75, 6);
  \draw[orange!70!white, thin] (8.5, 0) -- ++(0, 6) -- (10.75, 6) -- ++(0, -6);

  \fill[orange!35!white] (12.75, 0) rectangle (14.875, 6);
  \draw[orange!70!white, thin] (12.75, 0) -- ++(0, 6) -- (14.875, 6) -- ++(0, -6);

  \xguide[$x$]{4}{0}
  \yguide[$w_x$]{0}{6}

  \axes{16}{8.5}{$0$}{age~$a$}{$0$}{rank~$r(a)$}

  \draw[primary]
  \snake{(0, 2)}{(2.5, 6)}
  -- \snake{(2.5, 6)}{(3.75, 3)}
  -- \snake{(3.75, 3)}{(4.75, 6)}
  -- \snake{(4.75, 6)}{(5.5, 5)}
  -- \snake{(5.5, 5)}{(8, 7)}
  -- \snake{(8, 7)}{(9, 5)}
  -- \snake{(9, 5)}{(10, 4)}
  -- \snake{(10, 4)}{(11.5, 8)}
  -- \snake{(11.5, 8)}{(14, 4)}
  -- \snake{(14, 4)}{(15.75, 8)}
  -- \snakefirst{(15.75, 8)}{(17.25, 10)};

  \draw[secondary, densely dotted] (0, 6) -- ++(16.5, 0);
\end{tikzpicture}

%% file: figures/w_x-intervals_zeta.tex
\begin{tikzpicture}[figure]


  \newcommand{\wxintervalrectangle}[4][-2]{%
    \fill[orange!35!white] (#3, 0) rectangle (#4, 2);
    \draw[orange!70!white, thin] (#3, 0) -- ++(0, 2) -- (#4, 2) -- ++(0, #1);
    \draw [decorate, decoration={brace, raise=2pt, amplitude=4pt}]
      (#3 + 0.125, 2 + 0.5) -- (#4 - 0.125, 2 + 0.5)
      node [midway, above=0.5em] {$O(b_{#2})$};}

  \wxintervalrectangle{1}{2 + 0.01*1/2 + 0.125*1/2}{5 - 0.01*1/4 - 0.125*3/4}
  \wxintervalrectangle{2}{5 + 0.01*1/4 + 0.125*3/4}{11 - 0.01*1/8 - 0.125*7/8}
  \begin{scope}[]
    \clip (11, -1) rectangle (16.5, 9);
    \wxintervalrectangle{3}{11 + 0.01*1/8 + 0.125*7/8}{18}
  \end{scope}

  \xguide[$x$]{1.1}{0}
  \xguide[$b_1$]{2 + 0.01*1/2 + 0.125*1/2}{0}
  \xguide[$b_2$]{5 + 0.01*1/4 + 0.125*3/4}{0}
  \xguide[$b_3$]{11 + 0.01*1/8 + 0.125*7/8}{0}
  \yguide[$w_x$]{0}{2}

  \axes{16}{8.5}{$0$}{age~$a$}{$0$}{rank~$r(a)$}

  \begin{scope}[shift={(0, 1)}]
    \draw[primary]
      (0, 0)
      -- (0.5 - 0.125, 0) -- (0.5 - 0.01, 1) -- (0.5 + 0.01, 1) -- (0.5 + 0.125, 0)
      -- (2 - 0.125, 0) -- (2 - 0.01, 2) -- (2 + 0.01, 2) -- (2 + 0.125, 0)
      -- (5 - 0.125, 0) -- (5 - 0.01, 4) -- (5 + 0.01, 4) -- (5 + 0.125, 0)
      -- (11 - 0.125, 0) -- (11 - 0.01, 8) -- (11 + 0.01, 8) -- (11 + 0.125, 0) 
      -- (16.5, 0);
  \end{scope}

  \draw[secondary, densely dotted] (0, 2) -- ++(16.5, 0);
\end{tikzpicture}

%% file: figures/w_x-intervals_theta.tex
\begin{tikzpicture}[figure]


  \newcommand{\wxintervalrectangle}[1]{%
    \fill[orange!35!white] (#1, 0) rectangle ++(1.0, 2.0);
    \draw[orange!70!white, thin] (#1, 0) -- ++(0, 2.0) -- ++(1.0, 0) -- ++(0, -2.0);
    \draw [decorate, decoration={brace, raise=2pt, amplitude=4pt}]
      (#1 + 0.125, 2.0 + 0.5) -- ++(1.0 - 0.25, 0)
      node [midway, above=0.5em] {$O(x)$};}

  \wxintervalrectangle{2.5}
  \wxintervalrectangle{6.5}
  \wxintervalrectangle{14.5}

  \xguide[$x$]{1.1}{0}
  \yguide[$w_x$]{0}{2.0}

  \axes{16}{8.5}{$0$}{age~$a$}{$0$}{rank~$r(a)$}

  \begin{scope}[shift={(0, 1)}]
    \draw[primary]
      (0, 0) -- (0.5, 1) -- (1, 0) -- (2, 2) -- (3, 0)
      -- (5, 4) -- (7, 0) -- (11, 8) -- (15, 0) -- (16.5, 3);
  \end{scope}

  \draw[secondary, densely dotted] (0, 2.0) -- ++(16.5, 0);
\end{tikzpicture}

%% file: figures/light_gittins_repair.tex
\begin{tikzpicture}[figure]
  \newcommand{\aTilde}{8.50553}

  \xguide[$\tilde{a}$]{\aTilde}{20/3}
  \yguide[$1/\mu_1$]{0}{2.28455524}
  \yguide[$1/\mu_2$]{16.5}{8}

  \axes{16}{8.5}{$0$}{age~$a$}{$0$}{rank~$r_\gittins(a)$}

  \draw[primary, opacity=0.5]
    plot[domain=0:16.5, variable=\x, smooth]
      ({\x}, {2 + 6 / (1 + exp(-0.5 * (\x - 6))});
  \draw[tertiary, line width=2.4pt, densely dashed]
    plot[domain=0:\aTilde, variable=\x, smooth]
      ({\x}, {1.2 * (2 + 6 / (1 + exp(-0.5 * (\x - 6))))})
    -- plot[domain=\aTilde:16.5, variable=\x, smooth]
      ({\x}, {2 + 6 / (1 + exp(-0.5 * (\x - 6)))});
  \point[tertiary, fill=white]{(\aTilde, 20/3)}
  \point[tertiary]{(\aTilde, 8)}
\end{tikzpicture}

%% file: figures/yz.tex
\begin{tikzpicture}[figure]
  \xguide[$y_x$]{2.5}{6}
  \xguide[$x$]{4}{0}
  \xguide[$z_x$]{6.75}{6}
  \yguide[$w_x$]{16.5}{6}

  \axes{16}{8.5}{$0$}{age~$a$}{$0$}{rank~$r(a)$}

  \draw[primary]
  \snake{(0, 2)}{(2.5, 6)}
  -- \snake{(2.5, 6)}{(3.75, 3)}
  -- \snake{(3.75, 3)}{(4.75, 6)}
  -- \snake{(4.75, 6)}{(5.5, 5)}
  -- \snake{(5.5, 5)}{(8, 7)}
  -- \snake{(8, 7)}{(9, 5)}
  -- \snake{(9, 5)}{(10, 4)}
  -- \snake{(10, 4)}{(11.5, 8)}
  -- \snake{(11.5, 8)}{(14, 4)}
  -- \snake{(14, 4)}{(15.75, 8)}
  -- \snakefirst{(15.75, 8)}{(17.25, 10)};
\end{tikzpicture}

%% file: figures/heavy_gittins_goal.tex
\begin{tikzpicture}[figure]
  \fill[orange!35!white] (0, 0) rectangle (6.75, 6);
  \draw[orange!70!white, thin] (0, 6) -- (6.75, 6) -- ++(0, -6);
  \draw [decorate, decoration={brace, raise=2pt, amplitude=4pt}]
    (0.125, 6.5) -- (6.75 - 0.125, 6.5)
    node [midway, above=0.5em] {$O(x)$};

  \fill[orange!35!white] (8.5, 0) rectangle (10.75, 6);
  \draw[orange!70!white, thin] (8.5, 0) -- ++(0, 6) -- (10.75, 6) -- ++(0, -6);
  \draw [decorate, decoration={brace, raise=2pt, amplitude=4pt}]
    (8.5 + 0.125, 6.5) -- (10.75 - 0.125, 6.5)
    node [midway, above=0.5em] {$O(x)$};

  \fill[orange!35!white] (12.75, 0) rectangle (14.875, 6);
  \draw[orange!70!white, thin] (12.75, 0) -- ++(0, 6) -- (14.875, 6) -- ++(0, -6);
  \draw [decorate, decoration={brace, raise=2pt, amplitude=4pt}]
    (12.75 + 0.125, 6.5) -- (14.875 - 0.125, 6.5)
    node [midway, above=0.5em] {$O(x)$};

  \xguide[$x$]{4}{0}
  \yguide[$w_x$]{0}{6}

  \axes{16}{8.5}{$0$}{age~$a$}{$0$}{rank~$r_\gittins(a)$}

  \draw[primary]
  \snake{(0, 2)}{(2.5, 6)}
  -- \snake{(2.5, 6)}{(3.75, 3)}
  -- \snake{(3.75, 3)}{(4.75, 6)}
  -- \snake{(4.75, 6)}{(5.5, 5)}
  -- \snake{(5.5, 5)}{(8, 7)}
  -- \snake{(8, 7)}{(9, 5)}
  -- \snake{(9, 5)}{(10, 4)}
  -- \snake{(10, 4)}{(11.5, 8)}
  -- \snake{(11.5, 8)}{(14, 4)}
  -- \snake{(14, 4)}{(15.75, 8)}
  -- \snakefirst{(15.75, 8)}{(17.25, 10)};

  \draw[secondary, densely dotted] (0, 6) -- ++(16.5, 0);
\end{tikzpicture}

%% file: figures/ybcz.tex
\begin{tikzpicture}[figure]
  \fill[orange!35!white] (12.75, 0) rectangle (14.875, 6);
  \draw[orange!70!white, thin] (12.75, 0) -- ++(0, 6) -- (14.875, 6) -- ++(0, -6);

  \xguide[$x$]{4}{0}
  \yguide[$w_x$]{0}{6}

  \xguide[$u$]{13.75}{0}
  \xguide[$y_u$]{11.5}{8}
  \xguide[$b$]{12.75}{0}
  \xguide[$c$]{14.875}{0}
  \xguide[$z_u$]{15.75}{8}
  \yguide[$w_u$]{16.5}{8}

  \axes{16}{8.5}{$0$}{age~$a$}{$0$}{rank~$r(a)$}

  \draw[primary]
  \snake{(0, 2)}{(2.5, 6)}
  -- \snake{(2.5, 6)}{(3.75, 3)}
  -- \snake{(3.75, 3)}{(4.75, 6)}
  -- \snake{(4.75, 6)}{(5.5, 5)}
  -- \snake{(5.5, 5)}{(8, 7)}
  -- \snake{(8, 7)}{(9, 5)}
  -- \snake{(9, 5)}{(10, 4)}
  -- \snake{(10, 4)}{(11.5, 8)}
  -- \snake{(11.5, 8)}{(14, 4)}
  -- \snake{(14, 4)}{(15.75, 8)}
  -- \snakefirst{(15.75, 8)}{(17.25, 10)};

  \draw[secondary, densely dotted] (0, 6) -- ++(16.5, 0);
\end{tikzpicture}

%% file: figures/step.tex
\begin{tikzpicture}[figure]

  \xguide[$a^*$]{3}{0}
  \yguide[$r^*$]{3}{5}

  \axes{6.5}{7.5}{$0$}{age~$a$}{$0$}{rank~$r_\step(a)$}

  \draw[primary] (0, 0) -- (3, 0) (3, 5) -- (7, 5);
  \draw[primary] (3, 0) -- (3, 5);
  \jump{3}{0}{5}

\end{tikzpicture}

%% file: figures/spike.tex
\begin{tikzpicture}[figure]


  \xguide[$a^*$]{3}{0}
  \yguide[$r^*$]{3}{5}

  \axes{6.5}{7.5}{$0$}{age~$a$}{$0$}{rank~$r_\spike(a)$}

  \draw[primary] (0, 0) -- (7, 0);
  \draw[primary] (3, 0) -- (3, 5);
  \jump[primary]{3}{0}{5}

\end{tikzpicture}

%% file: figures/generic_intermediate.tex
\begin{tikzpicture}[figure]

  \xguide[$a^*$]{3}{5}
  \yguide[$r^*$]{3}{5}

  \axes{6.5}{7.5}{$0$}{age~$a$}{$0$}{rank~$r_\pi(a)$}

  \draw[primary] \snake{(0, 1)}{(3, 5)} \snake{(3, 5)}{(6, 2)} \snakefirst{(6, 2)}{(8, 3.5)};;

\end{tikzpicture}

%% file: figures/sigma.tex
\begin{tikzpicture}[figure]
  \fliplrmode

  \xguide[$s_0$]{-8}{12}
  \xguide[$s_1$]{-6}{0}
  {\flipudmode\xguide[$s_2$]{-4}{-2}}
  \yguide[$s_3$]{-4}{-2}

  \draw[axis, <->] ({-9 - \xarrowsize}, 0) -- ({5 + \xarrowsize}, 0) node[right] {$s$};
  \draw[axis, <->] (0, {-3.5 - \yarrowsize}) -- (0, {11.5 + \yarrowsize}) node[above] {$\sigma^{-1}(s)$};

  \draw[orange!70!white, opacity=0.5, line width=4.2pt]
    plot[domain=-4:5.5, variable=\x, smooth] ({\x}, {\x - 2*(1 - 8/(8+\x))});

  \draw[tertiary]
    plot[domain=-7.25:-4, variable=\x, smooth] ({\x}, {\x - 2*(1 - 8/(8+\x))})
    plot[domain=-4:5.5, variable=\x, smooth] ({\x}, {\x - 2*(1 - 8/(8+\x))});

  \node at (18, 4.25) {$\begin{aligned}
    s_0 &= \gamma(\sigma^{-1}) = \gamma(\lst{X}) \\
    s_1 &= \gamma(\lst{W}) = \text{least root of $\sigma^{-1}$} \\
    s_2 &= \sigma(\gamma(\sigma)) = \argmin_{s \geq s_0} \sigma^{-1}(s) \\
    s_3 &= \gamma(\sigma) = \min_{s \geq s_0} \sigma^{-1}(s)
  \end{aligned}$};
\end{tikzpicture}

%% file: figures/step_worst.tex
\begin{tikzpicture}[figure]

  \xguide[$a^*$]{3}{0}
  \yguide[$r^*$]{0}{5}
  \xguide[$x$]{4.5}{5}

  \axes{6.5}{7.5}{$0$}{age~$a$}{$0$}{w.f.r.~$w_{\step, x}(a)$}

  \draw[primary, opacity=0.5] (0, 0) -- (3, 0) (3, 5) -- (7, 5);
  \draw[primary, opacity=0.5] (3, 0) -- (3, 5);
  \jump{3}{0}{5}

  \draw[secondary, line width=2.4pt, densely dashed] (0, 5) -- (4.5, 5);

\end{tikzpicture}

%% file: figures/spike_worst.tex
\begin{tikzpicture}[figure]

  \xguide[$a^*$]{3}{0}
  \yguide[$r^*$]{0}{5}
  \xguide[$x$]{4.5}{0}

  \axes{6.5}{7.5}{$0$}{age~$a$}{$0$}{w.f.r.~$w_{\spike, x}(a)$}

  \draw[primary, opacity=0.5] (0, 0) -- (7, 0);
  \draw[primary, opacity=0.5] (3, 0) -- (3, 5);

  \draw[secondary, line width=2.4pt, densely dashed] (0, 5) -- (3, 5) (3, 0) -- (4.5, 0);
  \draw[secondary, line width=2.4pt, densely dashed] (3, 0) -- (3, 5);
  \jump[secondary]{3}{0}{5}

\end{tikzpicture}

%% file: figures/generic_intermediate_worst.tex
\begin{tikzpicture}[figure]

  \xguide[$a^*$]{3}{5}
  \yguide[$r^*$]{0}{5}
  \xguide[$x$]{4.5}{3.5}

  \axes{6.5}{7.5}{$0$}{age~$a$}{$0$}{w.f.r. $w_{\pi, x}(a)$}

  \draw[primary, opacity=0.5] \snake{(0, 1)}{(3, 5)} \snake{(3, 5)}{(6, 2)} \snakefirst{(6, 2)}{(8, 3.5)};
  \draw[secondary, line width=2.4pt, densely dashed] (0, 5) -- \snakefirst{(3, 5)}{(6, 2)};

\end{tikzpicture}